\documentclass[journal,10pt,onecolumn,draftclsnofoot,]{IEEEtran}
\usepackage{bm}
\usepackage{cite}
\usepackage{amsmath,amssymb,amsfonts}
\usepackage{algorithmic}
\usepackage{graphicx}
\usepackage{textcomp}
\usepackage{dsfont}

\usepackage{scalerel}
\usepackage{tikz}
\usetikzlibrary{svg.path}

\definecolor{subsectioncolor}{rgb}{0,0.541,0.855}

\definecolor{orcidlogocol}{HTML}{A6CE39}
\tikzset{
    orcidlogo/.pic={
        \fill[orcidlogocol] svg{M256,128c0,70.7-57.3,128-128,128C57.3,256,0,198.7,0,128C0,57.3,57.3,0,128,0C198.7,0,256,57.3,256,128z};
        \fill[white] svg{M86.3,186.2H70.9V79.1h15.4v48.4V186.2z}
        svg{M108.9,79.1h41.6c39.6,0,57,28.3,57,53.6c0,27.5-21.5,53.6-56.8,53.6h-41.8V79.1z M124.3,172.4h24.5c34.9,0,42.9-26.5,42.9-39.7c0-21.5-13.7-39.7-43.7-39.7h-23.7V172.4z}
        svg{M88.7,56.8c0,5.5-4.5,10.1-10.1,10.1c-5.6,0-10.1-4.6-10.1-10.1c0-5.6,4.5-10.1,10.1-10.1C84.2,46.7,88.7,51.3,88.7,56.8z};
    }
}

\newcommand\orcidicon[1]{\href{https://orcid.org/#1}{\mbox{\scalerel*{
                \begin{tikzpicture}[yscale=-1,transform shape]
                \pic{orcidlogo};
                \end{tikzpicture}
            }{|}}}}
            
\def\BibTeX{{\rm B\kern-.05em{\sc i\kern-.025em b}\kern-.08em
    T\kern-.1667em\lower.7ex\hbox{E}\kern-.125emX}}

\makeatletter
\def\fnum@figure{\textcolor{subsectioncolor}{\sf Fig.~\thefigure}}
\def\fnum@table{\textcolor{subsectioncolor}{\sf TABLE~\thetable}}
\makeatother

\usepackage[hidelinks]{hyperref}

\begin{document}
\title{Scheduling Flexible Non-Preemptive Loads in Smart-Grid Networks}
\author{Nathan Dahlin and Rahul Jain
\thanks{Submitted November 5, 2020. This work was supported by NSF Awards ECCS-1611574 and ECCS-1810447.}
\thanks{Nathan Dahlin and Rahul Jain are with the Department of Electrical and Computer Engineering, University of Southern California, 3740 McClintock Ave, Los Angeles, CA, 90089, USA (e-mail: dahlin@usc.edu, rahul.jain@usc.edu, phone: (213) 631 6101).}
\thanks{Correspondence:  3740 McClintock Ave, Los Angeles, CA, 90089, USA }
}

\newtheorem{theorem}{Theorem}
\newtheorem{lemma}[theorem]{Lemma}
\newtheorem{proposition}[theorem]{Proposition}
\newtheorem{corollary}{Corollary}[theorem]
\newtheorem{definition}{Definition}
\newtheorem{example}{Example}
\newtheorem{remark}{Remark}
\newtheorem{assumption}{Assumption}
\newtheorem{proof}{Proof}
\renewcommand{\theproof}{\unskip}

\maketitle

\begin{abstract}
A market consisting of a generator with thermal and renewable generation capability, a set of \textit{non-preemptive} loads (i.e., loads which cannot be interrupted once started), and an independent system operator (ISO) is considered. Loads are characterized by durations, power demand rates and utility for receiving service, as well as disutility functions giving preferences for time slots in which service is preferred. Given this information, along with the generator's thermal generation cost function and forecast renewable generation, the social planner solves a mixed integer program to determine a load activation schedule which maximizes social welfare. Assuming price taking behavior, we develop a competitive equilibrium concept based on a relaxed version of the social planner's problem which includes prices for consumption and incentives for flexibility, and allows for probabilistic allocation of power to loads. Considering each load as representative of a population of identical loads with scaled characteristics, we demonstrate that the relaxed social planner's problem gives an exact solution to the original mixed integer problem in the large population limit, and give a market mechanism for implementing the competitive equilibrium. Finally, we evaluate via case study the benefit of incorporating load flexibility information into power consumption and generation scheduling in terms of proportion of loads served and overall social welfare.
\end{abstract}

\begin{IEEEkeywords}
Power systems economics, power system planning, electric vehicles
\end{IEEEkeywords}

\section{Introduction}
\label{sec:introduction}
Over the roughly century long history of the electrical power grid, the situation facing both grid managers and end users has remained largely the same: electricity available on demand. In the case of the latter, operation of lightbulbs, television sets and other appliances has been just the flip of a switch away, while for the former, the set of available controls and actions was over supply, i.e., which generators to activate - how much to generate and when \cite{changingwhenweuseenergy}? Managers have consistently succeeded in providing an adequate supply to meet the demand of end users from second to second largely due to the fact that over the past century, demand forecasting has reached day ahead accuracy within 5\% \cite{forecasterror}. 

Recently, circumstances have changed on both the supply and demand sides of the grid. Increased adoption of renewables means that the available power supply is becoming less controllable. Thus, even in the presence of relatively predictable aggregate load, forecasting errors in \textit{excess load} can be significant. Meanwhile, the rise of networked appliances, homes and buildings is now facilitating synchronization and coordination of consumption to the extent that the demand side flexibility stands to become one of the most important assets available to grid operators \cite{qin2018automatic}. Water heaters and electric vehicles (EV) typify loads characterized by such flexibility. A newly published report from the Brattle Group estimates that load flexibility could be expanded to satisfy nearly 20 percent of US peak demand, and avoid nearly \$18 billion in annual generation capacity, energy, transmission and ancillary service costs \cite{brattleflexibility}.

Currently, aggregate flexibility is leveraged through \textit{demand response} programs. Typically these programs are used to reduce peaks in demand, either by indirect load control via real-time pricing or direct control, where utilities have the ability to turn devices on or off. Moving forward, much of the additional benefit is expected to come from expanding the use of demand response to applications such as load shifting and building, e.g., to track a time varying supply of renewable energy, and services such as frequency regulation and voltage control \cite{brattleflexibility}.

This work considers a population of \textit{non-preemptive} loads, i.e., loads which must be served continuously for a predetermined amount of time without interruption once service has started. Examples of such loads are household appliances like dishwashers, and EV charging with tight deadlines \cite{hashmi2018load}. Users report their level of discomfort for being served at each time slot of a finite time horizon. The social planner is tasked with serving these loads has access to a thermal generator with convex generation cost, as well as a renewable generator with zero marginal cost. Given the users' preferences, thermal generator's cost function, and knowledge of the renewable generator's output, the scheduler determines an efficient schedule for cost minimization. We seek to answer the following questions: How can these flexible loads be scheduled over the available time slots? Once a schedule has been determined, how should users be compensated for their flexibility? What is the ``price of inflexibility'' in this setting? In particular, the problem of monetizing flexibility has proven quite challenging thus far due to a lack of suitable optimization formulations \cite{qin2018automatic}.


While the problem of scheduling processor time for service of both preemptive and non-preemptive tasks has long been studied in the computer science literature \cite{gupta2015scheduling}, load scheduling in the context of demand response in energy systems has received considerable attention over the past several years. Model predictive control (MPC) \cite{habib2016model}, successive binary optimization \cite{sun2016optimal} and greedy algorithms \cite{o2015scheduling} have been applied to the particular objective of tracking a target aggregate load profile at minimum cost, without considering pricing.

Beyond costs associated directly with generation, several works incorporate exogenous energy pricing schemes as parameters of their respective formulations. Fixed uniform rates \cite{subramanian2013real}, tiered rates based upon given affine per unit pricing functions \cite{han2017optimal}, and peak/off-peak pricing schemes \cite{hijjo2018scheduling} modeling programs implemented by utilities have been considered in concert with earliest deadline first (EDF), least laxity first (LLF) and MPC scheduling policies. 

Of the works that derive flexibility pricing schemes endogenously alongside optimal schedules, relatively few consider non-preemptive loads. This is primarily due to the binary start time decision variables necessary when introducing non-preemptive loads, which precludes direct use of traditional convex optimization or marginal pricing based schemes \cite{qin2018automatic}. Without explicit inclusion of non-preemptive consumption profiles, the scheduling and pricing of a continuum of deadline differentiated loads is studied in \cite{bitar2016deadline}, wherein the longer a consumer is willing to defer, the lower their energy price. 

Targeting non-preemptive loads with hard constraints on acceptable service windows, both \cite{carrasqueira2017bi} and \cite{meng2014optimal} adopt a bilevel programming approach in which an energy provider optimizes a pricing schedule in the upper level, and consumers react by scheduling their loads in order to minimize costs and discomfort. 
Genetic algorithm \cite{meng2014optimal}, and evolutionary and particle-swarm \cite{carrasqueira2017bi} based scheduling heuristics are used to handle the non-preemptive related integer constraints.

More recently, \cite{qin2018automatic} details a power exchange platform allowing for 
flexible suppliers, as well as consumers with non-preemptive consumption patterns. A fluid relaxation on demand profile shapes, along with a projection method for deriving a feasible schedule is given, and the resulting solution is shown to be asymptotically optimal in the infinite load limit. 
Marginal pricing, given a schedule of the flexible loads is shown to be inadmissable with respect to incoming offers.

Closest to this work, in \cite{gupta2015scheduling}, a setting similar to the one presented here in continuous time is examined. Prices for load consumption and inflexibility are derived as dual variables to the scheduler's convex optimization problem, and a competitive equilibrium with respect to reported loads reported consumption level and duration is studied. Under time discretization, 
approximately optimal scheduling and pricing heuristics are developed. In contrast to this work, user disutility is not modeled, and flexibility incentives do not arise from the optimization formulation. Further, 
the optimality of the associated heuristics cannot be proven. 
\subsection{Statement of Contributions}

In this paper, we propose a tractable optimization formulation for scheduling and pricing non-preemptive load service. 
Inclusion of such loads necessitates specification of constraints and variables capturing their non-interruptible nature, and in general results in NP hard optimization problems \cite{qin2018automatic}. 
In summary, the key steps and contributions of this paper are as follows.
\begin{itemize}
\item We propose a novel specification of load flexibility and decentralized optimization formulation for the scheduling of non-preemptive loads.
\item We formulate a corresponding centralized welfare maximization problem, and prove the existence of a competitive equilibrium in the relaxed version of this setting with finitely many loads, i.e., we show that there exist prices for per unit energy consumption and inflexibility such that the thermal generator produces efficient levels at each time step, and a social planner schedules loads such that demand equals supply while respecting the loads' flexibility preferences
\item We prove that the competitive equilibrium determined for the finite population setting is also a competitive equilibrium in the original mixed binary problem, when each load is interpreted as representing an infinite population of loads with appropriately scaled demand. Thus, the prices derived via convex relaxation are suitable for use in the binary constrained setting. Such a result is currently absent in the related literature. 
\item We specify a market mechanism for implementing the competitive equilibrium.
\item We present a case study demonstrating the utility of our formulation, based on real world electric vehicle charging data drawn from the Adaptive Charging Network (ACN) project \cite{lee2019acn}.
\end{itemize} 

\section{Problem Formulation}
The market consists of $M$ non-preemptive loads (or consumers) and a single thermal generator. Additionally, an ISO (independent system operator) ensures safe grid operation. Let $\mathcal{T} =[1,\cdots,T]$ denote the discrete time horizon over which loads are scheduled and served. For simplicity, we assume a single bus network model. We assume throughout that all entities are \textit{price taking}, i.e., their actions do not affect market prices. 

\subsection{Market Entity Problems}
\subsubsection{Loads}

Each load $i$ is characterized by a tuple $(\tau_i,l_i,\overline{U}_i,u^{dS}_{i\cdot},u^{dE}_{i\cdot})$, where $\tau_i$ gives the duration in time slots, $l_i$ gives the consumption level, and $u^{dS}_{i\cdot}$ and $u^{dE}_{i\cdot}$ give the disutility functions of consumer $i$ due to service starting prior to or after a desired service window, respectively. Consumer $i$ demands $l_i$ MW of electricity for $\tau_i$ consecutive time slots, and derives utility $\overline{U}_i$ as their load is fulfilled. Figure 1 plots an example pair of disutility functions vectors. If consumer $i$ is \textit{allowed} to be served in time slot $t$, it suffers disutility $u^{dS}_{it}+u^{dE}_{it}\geq0$. For example, an EV commuter may prefer that their vehicle be charged for a five hour period at 5 kW/h between 9AM and 5PM - while they work - rather than arrive early or stay late in order to charge their vehicle. Thus, the consumer's overall utility is a function of the flexibility that it allows for in the scheduling of its load. 

\begin{figure}[htp]
\begin{center}
\includegraphics[width=3.5in]{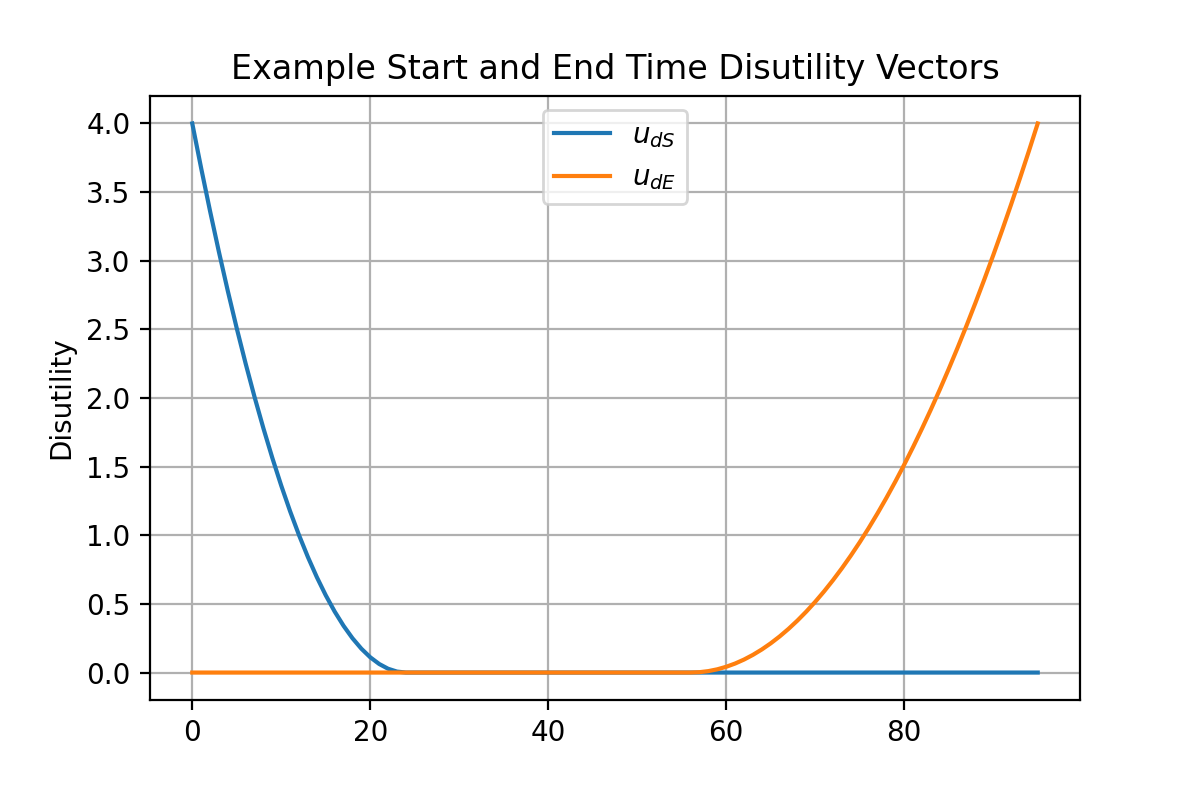}
\caption{Example disutility vectors.}
\label{StartEndDisutility}
\end{center}
\end{figure}

Denote $x^{C}_{i}:= (x^{C}_{i1},\dots,x^{C}_{iT})\in\{0,1\}^{T}$. We will similarly define vector and matrix valued quantities throughout. Given flexibility incentives $p^S_i\in\mathbb{R}^T_+$ and $p^E_i\in\mathbb{R}^T_+$, consumer $i$ solves the following optimization problem:
\begin{align}\nonumber(\text{CONS}_i)\quad \min_{\substack{x^C_i\in\{0,1\}^T,\\\,y^C_i\geq0,\,z^C_i\geq0}}&\quad\sum_tp^{\text{con}}_{it}x^C_{it} - \overline{U}_i\sum_{t=1}^{T-\tau_i+1}x^C_{it}+\sum_t\left((1-y^C_{it})(u^{dS}_{it}-p^S_{it}) +(1-z^C_{it})\left(u^{dE}_{it}-p^E_{it}\right)\right)\\
\label{CON_i_const1}\text{s.t.}&\quad \sum_{s=1}^{t}\sum_{r=\max\{1,s-\tau_i+1\}}^sx^C_{ir}\leq \tau_i(1-y^C_{it})\,\, \forall\,t\\
\label{CON_i_const2}&\quad\sum_{s=t}^{T}\sum_{r=\max\{1,s-\tau_i+1\}}^sx^C_{ir}\leq \tau_i(1-z^C_{it})\,\, \forall\,t. 
\end{align} 
In the context of EV charging, if $x^C_{it}=1$, then consumer $i$ chooses to begin charging their vehicle at time slot $t$ and pay activation price $p^{\text{con}}_{it}$. The inner sums on the left hand side in constraints (\ref{CON_i_const1}) and (\ref{CON_i_const2}) give the charging/idle status of load $i$ at each time slot $s$. That is, the sums will be equal to 1 if consumer $i$'s vehicle started charging in any time slot $\{\max\{1,s-\tau_i+1\},\dots, s\}$ and 0 otherwise. The term $(1-y^C_{it})=1$ when consumer $i$'s EV has started charging prior to or at time slot $t$. In such a case, consumer $i$ incurs disutility $u^{dS}_{it}\geq0$ for having started by time $t$, but is compensated at early start rate $p^{S}_{it}$. Similarly, $(1-z^C_{it})=1$ indicates that consumer $i$'s EV will be charging at or after time slot $t$, with $u^{dE}_{it}$ and late charge ending rate $p^{E}_{it}$ analogous to $u^{dS}_{it}$ and $p^S_{it}$. 

\subsubsection{Thermal Generator}
The generator is characterized by its thermal generation cost function $c(\cdot)\,:\,\mathbb{R}_+\to\mathbb{R}_+$, which is assumed to be strictly convex, increasing and twice differentiable on $\mathbb{R}_+$. In addition to the generator's thermal plant, we assume that it also owns a renewable generator which produces energy at zero marginal cost. The output of the renewable generator, $g\,:\,\mathcal{T}\to(0,\infty)$ is assumed to be known to all market participants at time $t=0$. Given prices $p^{\text{gen}}\in\mathbb{R}^T_+$, the generator chooses generation levels $q^G\in\mathbb{R}^T_+$ to solve the following profit maximization problem 
\begin{equation}\nonumber(\text{GEN})\quad \max_{q^G\geq0}\quad\sum_t\left(p^{\text{gen}}_t(q^G_t+g_t)-c(q^G_t)\right).\end{equation}

\subsubsection{ISO}
Finally, the ISO collects all load profiles and determines the set of admissible load and generation schedules by solving 
\begin{equation}\nonumber\begin{split}
(\text{ISO})\min_{\substack{q^I\geq0\\\,x^I\in\{0,1\}}}&\sum_tp^{\text{gen}}_t\bigg(q^I_t+g_t-\sum_il_i\sum_{s=\max\{1,t-\tau_i+1\}}x^I_{is}\bigg)\\
\text{s.t.}\quad &\,\, \sum_il_i\sum_{s=\max\{1,t-\tau_i+1\}}^tx^I_{is}-g_t\leq q^I_t\quad\forall\,t.
\end{split}\end{equation}
where $q^I\in\mathbb{R}^T_+$. Note that the ISO incurs positive cost at any time slot $t$ where thermal generation $q_t$ exceeds residual demand (aggregate demand less renewable generation), and thus will find those schedules which balance thermal generation and residual demand optimal.

\subsection{The Social Planner's Problem}

In order to study the welfare properties of the competitive equilibrium given later, we introduce a social planning problem. The social planner is concerned with maximizing the combined welfare of all market participants, while ensuring safe operation of the power grid. Specifically, the social planner collects the profiles of each load $i$, and schedules them so that each is served without interruption for their entire duration. In practice, either the ISO or equivalent market participant, or a government organization often assumes responsibility for these tasks \cite{elreview}. Here we introduce the social planner as a distinct entity for clarity as we investigate properties of our market formulation. Let $\hat{x}_{it}\in\{0,1\}$ denote the social planner's decision as to whether load $i$ will begin service in time slot $t$, where $\hat{x}_{it}=1$ denotes that load $i$ will start at time slot $t$. A schedule is then defined as $\hat{x}\in\{0,1\}^{M\times T}$. The social planner selects a schedule, auxiliary load status variables $\hat{y}$ and $\hat{z}$, and corresponding generation levels $\hat{q}:=(\hat{q}_1,\dots,\hat{q}_t)$ in order to solve the following problem:
\begin{align}
(\text{SPP})\min_{\substack{\hat{q}\geq 0\hat{x}\in\{0,1\}^{M\times T}\\\hat{y},\,\hat{z}\,\hat{y},\,\hat{z}}} &\sum_tc(\hat{q}_t)+ \sum_i\sum_tu^{dS}_{it}(1-\hat{y}_{it})+\sum_i\sum_tu^{dE}_{it}(1-\hat{z}_{it})- \sum_i\overline{U}_i\sum_{t=1}^{T-\tau_i+1}\hat{x}_{it}\\
\text{s.t.}\label{SPPconst11}\quad &\sum_il_i\sum_{s=\max\{1,t-\tau_i+1\}}^t\hat{x}_{is} -g_t\leq \hat{q}_t\quad\forall\,t\\
\label{SPPconst31}&\sum_{s=1}^{t}\sum_{r=\max\{1,s-\tau_i+1\}}^s\hat{x}_{ir}\leq \tau_i(1-\hat{y}_{it})\quad\forall\,i,\,t\\
&\sum_{s=t}^T\sum_{r=\max\{1,s-\tau_i+1\}}^s\hat{x}_{ir} \leq\tau_i(1-\hat{z}_{it})\quad\forall\,i,\,t.
\end{align}

\section{Convex Relaxation and Pricing}
In order to develop prices for electricity consumption and load inflexibility, as well as a competitive equilibrium concept, we relax the binary constraints on matrices $\hat{x}$, $\hat{y}$ and $\hat{z}$, and consider the following problem:
\begin{align}
(\text{SPP-R})\,\, \min_{\substack{\hat{q},\hat{x},\hat{y},\hat{z}\geq 0}}\,\,&\sum_tc(\hat{q}_t)+ \sum_i\sum_tu^{dS}_{it}(1-\hat{y}_{it})+\sum_i\sum_tu^{dE}_{it}(1-\hat{z}_{it})- \sum_i\overline{U}_i\sum_{t=1}^{T-\tau_i+1}\hat{x}_{it}\\
\label{SPPconst1}\text{s.t.}\quad &\hat{\lambda}_t\,:\,\sum_il_i\sum_{s=\max\{1,t-\tau_i+1\}}^t\hat{x}_{is}-g_t \leq \hat{q}_t\quad\forall\,t\\
\label{SPPconst3}\quad&\hat{\nu}^S_{it}\,:\,\sum_{s=1}^{t}\sum_{r=\max\{1,s-\tau_i+1\}}^s\hat{x}_{ir}\leq \tau_i(1-\hat{y}_{it})\quad\forall\,i,\,t\\
\label{SPPconst4}\quad &\hat{\nu}^E_{it}\,:\,\sum_{s=t}^T\sum_{r=\max\{1,s-\tau_i+1\}}^s\hat{x}_{ir} \leq\tau_i(1-\hat{z}_{it})\quad\forall\,i,\,t, 
\end{align}
where $\hat{\lambda}_t$, $\hat{\nu}^S_{it}$ and $\hat{\nu}^E_{it}$ denote the dual variables corresponding to constraints (\ref{SPPconst1}-\ref{SPPconst4}). 
It can be shown that constraints (\ref{SPPconst3}) and (\ref{SPPconst4}) ensure that all entries of $\hat{x}$, $\hat{y}$ and $\hat{z}$ are less than 1, and also that \begin{equation}\label{x_sum}\sum_{t=1}^{T-\tau_i+1}\hat{x}_{it}\leq 1\quad\forall\,i,\,t.\end{equation} 

Under the relaxation, since in addition to (\ref{x_sum}), each schedule decision variable $\hat{x}_{it}$ satisfies $0\leq \hat{x}_{it}\leq 1$, $\hat{x}$ may be interpreted as a matrix specifying the \textit{probability} that a given load of type $i$ will be scheduled at time slot $t$ for all $i$ and $t$. That is, for each $i$, the planner will choose $\hat{x}_{i\cdot}\in\mathbb{R}^T$ equal to $e_t$, the $t$th standard basis vector, with probability $\hat{x}_{it}$. Therefore, $\hat{x}$ in (SPP-R) gives a \textit{probabilistic} schedule for the loads and if, for a given $i$, (\ref{x_sum}) holds with equality, then load $i$ is certain to be activated at some time slot $t\in \mathcal{T}$. Otherwise, the load only has a chance of ever being activated. Fixing a matrix of probabilities $\hat{x}$, $(1-\hat{y}_{it})$ and $(1-\hat{z}_{it})$ give probabilities that load $i$ has been activated up to time $t$, and will be active from time slot $t$ onward, respectively, for all $i$ and $t$. Consequently, the (SPP-R) objective may be viewed as the expectation of overall social welfare, and the constraints as being met in expectation. This interpretation is key to the competitive equilibrium definition and properties we detail in later sections. 

Note that due to the nonnegativity of $u^{dS}_{it}$ and $u^{dE}_{it}$ for all $i,\,t$, for any fixed $\hat{x}$, it is always optimal to choose each entry of matrices $\{\bm{1}-\hat{y}\}$ and $\{\bm{1}-\hat{z}\}$ as small as possible, where $\bm{1}$ denotes the matrix of size $N\times T$ with each entry equal to 1. Therefore, constraints (\ref{SPPconst3}) and (\ref{SPPconst4}) may be replaced with equalities, and matrices $\hat{y}$ and $\hat{z}$ are completely determined given a particular $\hat{x}$. 
 
Having relaxed the binary constraints on matrices $\hat{x}$, $\hat{y}$ and $\hat{z}$, we may employ Lagrangian analysis in order to arrive at a solution to (SPP-R). Problem (SPP-R)'s Lagrangian is given by
\begin{equation}\nonumber
\begin{split}
\mathcal{L} = \sum_tc(\hat{q}_t) &+ \sum_i\sum_tu^{dS}_{it}(1-\hat{y}_{it}) +\sum_i\sum_tu^{dE}_{it}(1-\hat{z}_{it})-\sum_i\overline{U}_i\sum_{t=1}^{T-\tau_i+1}\hat{x}_{it}\\
&+\sum_t\hat{\lambda}_t\bigg(\sum_il_i\sum_{s=\max\{1,t-\tau_i+1\}}^t\hat{x}_{is}-g_t-\hat{q}_t\bigg)\end{split}\end{equation}
\begin{equation}\nonumber\begin{split}
&+\sum_i\sum_t\hat{\nu}^S_{it}\bigg(\sum_{s=1}^{t}\sum_{r=\max\{1,s-\tau_i+1\}}^s\hat{x}_{ir}-\tau_i(1-\hat{y}_{it})\bigg)\\
&+\sum_i\sum_t\hat{\nu}^E_{it}\bigg(\sum_{s=t}^T\sum_{r=\max\{1,s-\tau_i+1\}}^s\hat{x}_{ir} -\tau_i(1-\hat{z}_{it})\bigg).
\end{split}
\end{equation}
Let 
\begin{equation}\label{p_nu_lambda_def}\begin{split}p^{\hat{\lambda}}_{it} &= l_i\sum_{s=t}^{\min\{T,t+\tau_i-1\}}\hat{\lambda}_s\end{split}\end{equation}
\begin{equation}\label{p_nu_lambda_def1}\begin{split}
p^{\hat{\nu}}_{it} &= \sum_{s=t}^{T}\hat{\nu}^S_{is}\min\{s-t+1,\tau_i\}+ \sum_{s=1}^{\min\{T,t+\tau-1\}}\hat{\nu}^E_{is}\min\{T-t+1,\tau_i,\tau_i-(s-t)\}.
\end{split}\end{equation}
(See Appendix \ref{p_lambda_nu_deriv} for the derivation of $p^{\hat{\lambda}}$ and $p^{\hat{\nu}}$). 
Then, the (SPP-R) Lagrangian can be rearranged as 
\begin{equation}
\label{second_Lagrangian}\begin{split}
\mathcal{L} = \sum_t\bigg(c(\hat{q}_t)-\hat{\lambda}_t(\hat{q}_t+g_t)\bigg)&+\sum_i\sum_t\bigg(p^{\hat{\lambda}}_{it}+p^{\hat{\nu}}_{it}\bigg)\hat{x}_{it}- \sum_i\overline{U}_i\sum_{t=1}^{T-\tau_i+1}\hat{x}_{it}+ \sum_i\sum_t(1-\hat{y}_{it})(u^{dS}_{it}-\hat{\nu}^S_{it}\tau_i)\\
&+\sum_i\sum_t(1-\hat{z}_{it})(u^{dE}_{it}-\hat{\nu}^E_{it}\tau_i),
\end{split}
\end{equation}
and in addition to feasibility, the KKT optimality conditions for $(\text{SPP-R})$ are 
\begin{align}
\label{SPPRKKT1}c'(\hat{q}^*_t) - \hat{\lambda}^*_t&\geq 0\quad\forall\,t\\
\label{SPPRKKT2}\hat{q}^*_t\left(c'(\hat{q}^*_t) - \hat{\lambda}^*_t\right)&= 0\quad\forall\,t\\
\label{SPPRKKT3}p^{\hat{\lambda}^*}_{it}+p^{\hat{\nu}^*}_{it} - \overline{U}_i&\geq 0\quad\forall\,i,\,t\leq T-\tau_i+1\\
\label{SPPRKKT4}\hat{x}^*_{it}\left(p^{\hat{\lambda}^*}_{it}+p^{\hat{\nu}^*}_{it}- \overline{U}_i\right)&=0\quad\forall\,i,\,t\leq T-\tau_i+1
\end{align}
\begin{align}
\label{SPPRKKT3b}p^{\hat{\lambda}^*}_{it}+p^{\hat{\nu}^*}_{it}&\geq 0\quad\forall\,i,\,t> T-\tau_i+1\\
\label{SPPRKKT4b}\hat{x}^*_{it}\left(p^{\hat{\lambda}^*}_{it}+p^{\hat{\nu}^*}_{it}\right)&=0\quad\forall\,i,\,t> T-\tau_i+1\\
\label{SPPRKKT5}\tau_i\hat{\nu}^{S*}_{it}-u^{dS}_{it}&\geq 0\quad\forall\,i,\,t\\
\label{SPPRKKT6}\hat{y}^*_{it}\left(\tau_i\hat{\nu}^{S*}_{it}-u^{dS}_{it}\right)&= 0\quad\forall\,i,\,t\\
\label{SPPRKKT7}\tau_i\hat{\nu}^{E*}_{it}-u^{dE}_{it}&\geq 0\quad\forall\,i,\,t\\
\label{SPPRKKT8}\hat{z}^*_{it}\left(\tau_i\hat{\nu}^{E*}_{it}-u^{dE}_{it}\right)&= 0\quad\forall\,i,\,t\\
\label{SPPRKKT9}\hat{\lambda}^*_t\big(\sum_il_i\sum_{s=\max\{1,t-\tau_i+1\}}^t\hat{x}^*_{is}-g_t-\hat{q}^*_t\big)&=0\quad\forall\,t\\
\label{SPPRKKT11}\hat{\lambda}^*_t&\geq0\,\quad\forall\,i,\,t.
\end{align}
Note that due to condition (\ref{SPPRKKT4}), load $i$ may only be activated with nonzero probability during time slots when the sum $p^{\hat{\lambda}^*}_{it}+p^{\hat{\nu}^*}_{it}$ is equal to the constant marginal utility term $\overline{U}_i$. Additionally, condition (\ref{SPPRKKT2}) implies that for time slots in which it is optimal to produce a positive quantity of electricity, we have that $\hat{\lambda}^*_t=c'(\hat{q}^*_t)>0$, the marginal cost of production.  In turn, condition (\ref{SPPRKKT9}) implies that generated quantities of electricity in such time slots will be equal to demand less forecast renewable generation. 

In view of our interpretation of the (SPP-R) objective as the expected value of social welfare, and the constraints as being met in expectation, the solution to (SPP-R) yields a set of admissible load activation schedules which may be randomly selected by the social planner in a single shot of the original, binary constrained problem. We will further examine this correspondence in later sections. Fixing such an activation schedule and taking into account renewable generation output $g_t$, the optimal generation schedule follows from constraints (\ref{SPPconst11}) and (\ref{SPPconst1}), and condition (\ref{SPPRKKT8}).
\subsection{Consumer's Problem}
The second (SPP-R) Lagrangian expression (\ref{second_Lagrangian}) suggests the following decomposition of the relaxed social planner's problem into relaxed versions of the individual entity problems presented above. See \cite{dahlin2020scheduling} for the optimality conditions corresponding to each of these problems. Starting with the consumer problems, we have 
\begin{align}(\text{CONS-R}_i)\,\,\min_{\substack{x^C_i,\,y^C_i\\z^C_i\geq 0}}&\quad\sum_tp^{\text{con}}_{it}x^C_{it} - \overline{U}_i\sum_{t=1}^{T-\tau_i+1}x^C_{it} +\sum_t\left((u^{dS}_{it}-p^S_{it})(1-y^C_{it}) +\left(u^{dE}_{it}-p^E_{it}\right)(1-z^C_{it})\right)\\
\label{CONSR_const1}\text{s.t.}&\quad\theta^S_{it}\,:\,\,\sum_{s=1}^{t}\sum_{r=\max\{1,s-\tau_i+1\}}^sx^C_{ir}= \tau_i(1-y^C_{it})\quad \forall\,t\\
\label{CONSR_const2}&\quad\,\theta^E_{it}\,:\,\,\sum_{s=t}^{T}\sum_{r=\max\{1,s-\tau_i+1\}}^sx^C_{ir}= \tau_i(1-z^C_{it})\quad \forall\,t.
\end{align} 

Again, under the relaxation on the binary constraints on $x_i^C$, $y^C_i$ and $z^C_i$, we may interpret the consumer's problem as selecting probabilities of activation for each time slot $t$, in the interest of maximizing their expected net utility (here written in minimization form). 
\subsection{Generator and ISO Problems}
The generator's problem remains the same as before
\begin{align}
(\text{GEN-R})\quad \max_{q^G\geq0}&\quad \sum_t\left(p^{\text{gen}}_t(q^G_t+g_t)-c(q^G_t)\right). 
\end{align}
Finally, the relaxed ISO problem is given by 
\begin{align}
\nonumber(\text{ISO-R})\quad\min_{\substack{q^I\geq 0\\x^I\geq 0}}\quad&\sum_tp^{\text{gen}}_t\bigg(q^I_t+g_t-\sum_il_i\sum_{s=\max\{1,t-\tau_i+1\}}^tx_{is}^I\bigg)\\
\nonumber\text{s.t.}\quad&\alpha_t\,:\, \sum_il_i\sum_{s=\max\{1,t-\tau_i+1\}}^tx^I_{is}-g_t\leq q^I_t\quad\forall\,t,
\end{align}

See Appendix \ref{entity_optimality} for the optimality conditions for each of the individual entity optimization problems. 

\section{Competitive Equilibrium and Theorems of Welfare Economics}
The competitive or Walrasian equilibrium is a standard reference point in economic analysis for assessing market outcomes. A competitive equilibrium is specified by an allocation of goods and prices, with the defining characteristic that taking the equilibrium prices as given, every market participant finds it optimal to select the corresponding equilibrium allocation \cite{mas1995microeconomic}. At equilibrium prices, the quantity of goods demanded by consumers is equal to the quantity produced by suppliers., i.e., the market clears. Therefore, equilibrium prices provide a coordinating signal for markets to operate in a \emph{decentralized} fashion.

Assuming that competitive equilibrium exists for a particular market setting, it is natural to compare the equilibrium allocation to allocations which directly maximize the aggregate welfare of all market participants. The latter allocations are called \emph{efficient}, and in our setting are given by solutions to (SPP-R).

We now give the competitive equilibrium definition for our setting, and explore existence, as well as welfare properties of such an equilibrium. As related above, competitive equilibria are typically specified in two-sided settings involving consumers and producers maximizing their individual well being and profits, respectively. Similar to the analysis found in \cite{rossi2019interaction} and \cite{wannegkowshameysha11b}, we augment the standard definition to include a nonprofit entity, i.e., the ISO.
\begin{definition}\label{comp_eq}
  (Competitive Equilibrium). A tuple $(\overline{q}^*,\overline{x}^*,\overline{y}^*,\overline{z}^*,\overline{p}^{\text{con}*},\overline{p}^{S*},\overline{p}^{E*},\overline{p}^{\text{gen}*})$ with $\overline{p}^{\text{gen}*}\geq0$ is said to be a competitive equilibrium if, given $(\overline{p}^{\text{con}*}_i,\overline{p}^{S*}_i,\overline{p}^{E*}_i)$, $(\overline{x}^*_i,\overline{y}^*_i,\overline{z}^*_i)$ solves $(\text{CON-R}_i)$ for each $i$, $\overline{q}^*$ solves $(\text{GEN-R})$, given $\overline{p}^{\text{gen*}}$, $\overline{q}^*$ solves $(\text{GEN-R})$, and given $\overline{p}^{\text{gen}*}$, $(\overline{q}^*,\overline{x}^*)$ solves $(\text{ISO-R})$. 
  \end{definition}
  
  As noted in the previous section since solutions to $(\text{CON-R}_i)$ will, in general, give values of $x^C_{it}\in[0,1]$, the quantities $(\overline{x}^*,\overline{y}^*,\overline{z}^*)$ in the competitive equilibrium in Definition \ref{comp_eq} have probabilistic interpretations: consumers select probabilities $x^C_{it}$ of being scheduled at each time slot $t\in\mathcal{T}$, in order to maximize their \textit{expected} net utility. 
  
  Our first result addresses the existence of the competitive equilibrium defined above.
  \begin{theorem}\label{comp_eq_exists}
  There exists a competitive equilibrium, given by an optimal solution $(\hat{q}^*,\hat{x}^*,\hat{y}^*,\hat{z}^*)$ to $(\text{SPP-R})$, and the following prices derived from an optimal dual solution to $(\text{SPP-R})$
  \begin{equation}\label{priceselect}\begin{split}
\overline{p}^{\text{con}*}_{it}= p^{\hat{\lambda}*}_{it}  + p^{\hat{\nu}*}_{it},\quad\overline{p}^{\text{gen}*}_t  = \hat{\lambda}^*_t,\quad
\overline{p}^{S*}_{it} = \tau_i\hat{\nu}^{S*}_{it},\quad \overline{p}^{E*}_{it} = \tau_i\hat{\nu}^{E*}_{it},\quad   \end{split}\end{equation}
  for all $i$ and $t$. 
  \end{theorem}

\begin{proof}
  See Appendix \ref{thm_1_appendix}. 
\end{proof}

The two fundamental theorems of welfare economics describe the relationship between competitive equilibria and efficient allocations. The first fundamental theorem states that competitive equilibria lead to, or \emph{support} efficient allocations \cite{mas1995microeconomic}. The second fundamental theorem states that the converse also holds, and in our settings corresponds to Theorem 1. We now state and prove the first fundamental theorem for our setting, given by Theorem \ref{converse}. Whereas proofs of the efficiency of competitive equilibria often require that the balance of supply and demand be included in the definition of such equilibria, in our development this equality arises from the given formulation of the ISO's problem (ISO-R). That is, facing equilibrium prices, the ISO will act to balance supply and demand as it optimizes (ISO-R).

\begin{theorem}\label{converse}
Any competitive equilibrium forms an optimal solution for $(\text{SPP-R})$. 
\end{theorem}
\begin{proof}
By definition, the competitive equilibrium $(\overline{q}^*,\overline{x}^*,\overline{y}^*,\overline{z}^*,\overline{p}^{\text{con}*},\overline{p}^{S*},\overline{p}^{E*},\overline{p}^{\text{gen}*})$ satisfies 
\begin{align}
c'(\overline{q}^*_t) - \overline{p}^{\text{gen}*}_t&\geq 0\quad\forall\,t\\
\overline{q}^*_t\left(c'(\overline{q}^*_t) - \overline{p}^{\text{gen}*}_t\right)&= 0\quad\forall\,t\\
\overline{p}^{\text{con}*}_{it} + p^{\theta*}_{it}- \overline{U}_i &\geq 0\quad\forall\,i,\,t\leq T-\tau_i+1\\
\overline{x}^*_{it}\left(\overline{p}^{\text{con}*}_{it} + p^{\theta*}_{it}- \overline{U}_i\right)&= 0\quad\forall\,i,\,t\leq T-\tau_i+1\\
\overline{p}^{\text{con}*}_{it} + p^{\theta*}_{it}&\geq 0\quad\forall\,i,\,t> T-\tau_i+1\\
\overline{x}^*_{it}\left(\overline{p}^{\text{con}*}_{it}+ p^{\theta*}_{it}\right)&= 0\quad\forall\,i,\,t> T-\tau_i+1\\
\overline{p}^{S*}_{it}+\tau_i\theta^{S*}_{it}- u^{dS}_{it}&\geq 0\quad\forall\,i,\,t\\
\overline{y}^*_{it}\left(\overline{p}^{S*}_{it}+\tau_i\theta^{S*}_{it}- u^{dS}_{it}\right)&=0\quad\forall\,i,\,t\\
\overline{p}^{E*}_{it}+\tau_i\theta^{E*}_{it}-u^{dE}_{it}&\geq 0\quad\forall\,i,\,t\\
\overline{z}^*_{it}\left(\overline{p}^{E*}_{it}+\tau_i\theta^{E*}_{it}-u^{dE}_{it}\right)&=0\quad\forall\,i,\,t\\
\overline{p}^{\text{gen}*}_t-\alpha^*_t&\geq 0\quad\forall\,t\\
\overline{q}^*_t\left(\overline{p}^{\text{gen}*}_t-\alpha^*_t\right)&=0\quad\forall\,t\\
-p^{\overline{p}^{\text{gen}*}}_{it} + p^{\alpha^*}_{it}&\geq 0 \quad\forall\,t\\
\overline{x}^*_{t}\left(-p^{\overline{p}^{\text{gen}*}}_{it} + p^{\alpha^*}_{it}\right)&= 0 \quad\forall\,t\\
\alpha^*_t\left(\sum_il_i\sum_{s=\max\{1,t-\tau_i+1\}}^t\overline{x}^*_{is}-g_t- \overline{q}^*_t\right)&=0\quad\forall\,t\\
\alpha^*_t&\geq 0\quad\forall\,t
\end{align}
for some $\theta^{S*}$, $\theta^{E*}$ and $\alpha^*\geq 0$, as well as the feasibility conditions for each of the individual entity problems. 
Therefore, observing that for any $\overline{p}^{\text{gen}*}\geq0$ the form of the objective in $(\text{ISO-R})$ ensures that complementary slackness condition (\ref{SPPRKKT9}) will be satisfied at the competitive equilibrium, selecting $(\hat{q}^*,\hat{x}^*,\hat{y}^*,\hat{z}^*)=(\overline{q}^*,\overline{x}^*,\overline{y}^*,\overline{z}^*)$ as the primal variables, and dual variables $\hat{\lambda}^*=\overline{p}^{\text{gen}*}=\alpha^*$ and $(\hat{\nu}^{S*}_{it},\hat{\nu}^{E*}_{it})=(\overline{p}^{S*}/\tau_i + \theta^{S*}_{it},\overline{p}^{E*}/\tau_i+\theta^{E*}_{it})$ for all $i,\,t$, forms optimal primal and dual solutions to $(\text{SPP-R})$. 
\end{proof}

\section{Replicated and Large Economies}
In general a competitive equilibrium is not guaranteed to exist when the social planner's problem is a mixed integer programming problem. Nevertheless, our competitive equilibrium definition allows for probabilistic allocation to consumers, and thus the existence of a competitive equilibrium is related to the existence of a primal and dual solution to the (relaxed) (SPP-R) problem. In this section we justify the study of this relaxed problem by demonstrating its equivalence to the original, binary constrained (SPP) when each load $i$ is interpreted as representing an infinite population of identical loads, with scaled demand. 

Thus far, our development has crucially relied on the assumption that market participants are price taking, i.e., presented with market prices, they make decisions in view of their own preferences and constraints, revealing their true demand without consideration of how their choices might influence these prices. But why should they act in this manner? In economic theory, the notion of large economies provides one justification for adoption of this assumption. The essence of the argument is as follows. As the number of market participants increases, any influence that an individual participant might have on market prices diminishes. When when that number grows to infinite, that influence vanishes entirely, and the price taking assumption becomes reasonable \cite{aumann1964markets}.

Following this intuitive argument, the question of how to add individuals to the market still remains. A special method, known as \emph{replication}, is to introduce participants with preferences and constraints identical to existing types, in the same proportion as existing ones \cite{hildenbrand1970economies}. In the context of electric vehicle charging, this could mean scaling up the number of drivers with the same model vehicle and desired charging schedule, with demand scaled down so as to avoid infeasible aggregate demands as the population of each type grows. In general it can be shown that as participants are added in this way, those of the same type will receive the same allocation.

Suppose that each load $i$ is replicated $N$ times, and that the resulting loads have demand, utility and disutility scaled by $N$. 
Indexing the replicas of each type $i$ with the index $n$, the binary constrained SPP with $N$ replication is 
\begin{align}
\nonumber\hspace{-0.1in}(\text{SPP($N$)})\min_{\substack{\hat{x}\in\{0,1\}^{M\times N\times  T}\\\hat{q},\,\hat{y},\,\hat{z}\geq0}}&\sum_tc\left(\hat{q}_t\right)-\sum_i\sum_n\frac{\overline{U}_{i}}{N}\sum_{t=1}^{T-\tau_i+1}\hat{x}_{int}\nonumber\hspace{-0.1in}+\sum_i\sum_n\sum_t\frac{u^{dS}_{it}}{N}(1-\hat{y}_{int})+\sum_i\sum_n\sum_t\frac{u^{dE}_{it}}{N}(1-\hat{z}_{int})\\
\nonumber\text{s.t.}\quad&\hat{\lambda}_t\,:\,\sum_i\sum_n\frac{l_i}{N}\sum_{s=\max\{1,t-\tau_i+1\}}^t\hat{x}_{ins}-g_t\leq \hat{q}_t\quad\forall\,t\\
\label{const2_sppn}\quad &\hat{\nu}^S_{int}\,:\,\sum_{s=1}^t\sum_{r=\max\{1,s-\tau_i+1\}}^s\hat{x}_{inr}=\tau_i(1-\hat{y}_{int})\quad\forall\,i,\,n,\,t\\
\label{const3_sppn}\quad&\hat{\nu}^E_{int}\,:\,\sum_{s=t}^T\sum_{r=\max\{1,s-\tau_i+1\}}^s\hat{x}_{inr}=\tau_i(1-\hat{z}_{int})\quad\forall\,i,\,n,\,t.
\end{align}

We refer to the problem with $N$ replication which relaxes the binary constraint on $\hat{x}$ as SPP($N$)-R (instead of SPP(1)-R, we will still refer to the original relaxed problem as SPP-R). When we wish to emphasize the dependence of decision variables on the replication factor $N$, we will append $(N)$, e.g., $\hat{x}_{int}(N)$. 

\begin{proposition}
\label{Nsol}Let $(\hat{q}^*,\hat{x}^*,\hat{y}^*,\hat{z}^*,\hat{\lambda}^*,\hat{\nu}^{S*},\hat{\nu}^{E*})$ denote an optimal solution to SPP-R. Then for any $N$, an optimal solution to SPP($N$)-R can be formed by setting $\hat{x}^*_{int}(N) = \hat{x}^*_{it}$, $\hat{y}^*_{int}(N) = \hat{y}^*_{it}$, $\hat{z}^*_{int}(N) = \hat{z}^*_{it}$, for all $i,\,n,\,t$, $\hat{\lambda}^*_t(N) = \hat{\lambda}^*_t$ for all $t$, and $\hat{\nu}^{S*}_{int}(N)=\hat{\nu}^{S*}_{it}/N$ and $\hat{\nu}^{E*}_{int}(N)=\hat{\nu}^{E*}_{it}/N$ for all  $i$ and $t$. 
\end{proposition}

\begin{proof}
SPP($N$)-R has the following KKT conditions. For all $t$ \begin{align}
\label{SPPNRKKT1}c'(\hat{q}^*_t(N)) - \hat{\lambda}^*_t(N)&\geq 0\\
\label{SPPNRKKT2}\hat{q}^*_t(N)\left(c'(\hat{q}^*_t(N)) - \hat{\lambda}^*_t(N)\right)&= 0,
\end{align}
for all $i$, $n$ and $T\leq T-\tau_i+1$
\begin{align}
\label{SPPNRKKT3}\frac{p^{\hat{\lambda}*}_{it}(N)}{N} + p^{\hat{\nu}*}_{int}(N) -\frac{\overline{U}_i}{N}&\geq 0\\
\label{SPPNRKKT4}\hat{x}^*_{int}(N)\left(\frac{p^{\hat{\lambda}*}_{it}(N)}{N} + p^{\hat{\nu}*}_{int}(N) - \frac{\overline{U}_i}{N}\right)&= 0,
\end{align}
for all $i$, $n$, and $T> T-\tau_i+1$
\begin{align}
\label{SPPNRKKT5}\frac{p^{\hat{\lambda}*}_{it}(N)}{N} + p^{\hat{\nu}*}_{int}(N)&\geq 0\\
\label{SPPNRKKT6}\hat{x}^*_{int}(N)\left(\frac{p^{\hat{\lambda}*}_{it}(N)}{N} + p^{\hat{\nu}*}_{int}(N)\right)&= 0,
\end{align}
for all $i$, $n$, and $t$
\begin{align}
\label{SPPNRKKT7}\tau_i\hat{\nu}^{S*}_{int}(N)-\frac{u^{dS}_{it}}{N}&\geq 0\\
\label{SPPNRKKT8}\hat{y}^*_{int}(N)\left(\tau_i\hat{\nu}^{S*}_{int}(N)-\frac{u^{dS}_{it}}{N}\right)&= 0\\
\label{SPPNRKKT9}\tau_i\hat{\nu}^{E*}_{int}(N)-\frac{u^{dE}_{it}}{N}&\geq 0\\
\label{SPPNRKKT10}\hat{z}^*_{int}(N)\left(\tau_i\hat{\nu}^{E*}_{int}(N)-\frac{u^{dE}_{it}}{N}\right)&= 0,
\end{align}
and for all $t$
\begin{align}
\label{SPPNRKKT11}\hat{\lambda}^*_{t}(N)\bigg(\sum_i\frac{l_i}{N}\sum_n\sum_{s=\max\{1,t-\tau_i+1\}}^t\hat{x}^*_{ins}(N)-g_t- \hat{q}^*_t(N)\bigg)&=0 \\
\label{SPPNRKKT12}\hat{\lambda}^*_{t}(N)&\geq 0.
\end{align}
The proof of the theorem follows from making the selections specified in the theorem statement, substituting into (\ref{SPPNRKKT1})-(\ref{SPPNRKKT12}), and comparing with (\ref{SPPRKKT1})-(\ref{SPPRKKT11}).\end{proof}
%

Proposition \ref{Nsol} states that an optimal probabilistic schedule $\hat{x}^*(N)$ in the problem with $N$ replication can be derived from an optimal probabilistic schedule $\hat{x}^*$ for (SPP-R) and specifies how to do so. In the limit as $N\to \infty$ we can use $\hat{x}$ to generate an optimal \textit{deterministic}, binary constrained schedule if we interpret $\hat{x}^*_{it}$ as the proportion of the population of type $i$ to be activated at time $t$. This is stated formally in the following theorem. 
\begin{theorem}An optimal solution to SPP($\infty$) is given by activating proportion $\hat{x}^*_{it}$ of type $i$ population at time $t$ for each $i$ and $t$, where $\hat{x}^*$ is an optimal solution to SPP-R. 
\end{theorem}
\begin{proof}
Note that constraints (\ref{const2_sppn}) and (\ref{const3_sppn}) may be rewritten
  \begin{equation}\begin{split}
  \hat{y}_{int} &= 1-\frac{1}{\tau_i}\sum_{s=1}^t\sum_{r=\max\{1,s-\tau_i+1\}}^s\hat{x}_{inr}\quad \forall\,i,\,n,\,t\\
    \hat{z}_{int} &= 1-\frac{1}{\tau_i}\sum_{s=t}^T\sum_{r=\max\{1,s-\tau_i+1\}}^s\hat{x}_{inr}\quad \forall\,i,\,n,\,t, 
  \end{split}\end{equation}
  so that overall SPP($N$) can be written as 
  \begin{align}
\nonumber\min_{\substack{\hat{q}\geq0,\\\hat{x}\in\{0,1\}}} \quad&\sum_tc\left(\hat{q}_t\right)-\sum_i\overline{U}_{i}\sum_{t=1}^{T-\tau_i+1}\frac{1}{N}\sum_n\hat{x}_{int}+\sum_i\sum_tu^{dS}_{it}\bigg(\frac{1}{\tau_i}\sum_{s=1}^t\sum_{r=\max\{1,s-\tau_i+1\}}^s\frac{1}{N}\sum_n\hat{x}_{inr}\bigg)\\
\nonumber&\hspace{2.2in}+\sum_i\sum_tu^{dE}_{it}\bigg(\frac{1}{\tau_i}\sum_{s=t}^T\sum_{r=\max\{1,s-\tau_i+1\}}^s\frac{1}{N}\sum_n\hat{x}_{inr}\bigg)\\
\label{const1_sppn_proof}\text{s.t.}\quad &\hat{\lambda}_t\,:\,\sum_il_i\sum_{s=\max\{1,t-\tau_i+1\}}^t\frac{1}{N}\sum_n\hat{x}_{ins}-g_t\leq \hat{q}_t\quad\forall\,t.
\end{align}
Now, if $\hat{x}_{int}(N)$ is considered as a Bernoulli random variable with $P(\hat{x}_{int}(N)=1) = \hat{x}^*_{it}$ and $\hat{q}^*_t(N)$ is chosen as $\hat{q}^*_t(1)=\hat{q}_t$ for all $t$, then by the Law of Large Numbers, constraint (\ref{const1_sppn_proof}) converges to 
$$\sum_il_i\sum_{s=\max\{1,t-\tau_i+1\}}^t\hat{x}^*_{is}-g_t\leq \hat{q}^*_t\quad \forall\,t.$$
Similarly, the objective function converges to
\begin{equation}\nonumber\begin{split}
\sum_tc\left(\hat{q}^*_t\right)-\sum_i\overline{U}_{i}\sum_{t=1}^{T-\tau_i+1}\hat{x}^*_{it}&+\sum_i\sum_tu^{dS}_{it}\bigg(\frac{1}{\tau_i}\sum_{s=1}^t\sum_{r=\max\{1,s-\tau_i+1\}}^s\hat{x}^*_{ir}\bigg)\\
&+\sum_i\sum_tu^{dE}_{it}\bigg(\frac{1}{\tau_i}\sum_{s=t}^T\sum_{r=\max\{1,s-\tau_i+1\}}^s\hat{x}^*_{ir}\bigg).
\end{split}\end{equation}
Since the optimal objective of the relaxed problem provides a lower bound for the binary constrained problem, and the power balance constraint is satisfied in the limit as $N\to \infty$, the solution produced by randomly activating loads according to $\hat{x}^*$ converges to an optimal binary constrained solution as $N\to\infty$. 
\end{proof}

\section{Market Mechanism for Large Population Economy}

Market mechanism design is an approach in economic theory which, rather than taking economic institutions as fixed and predicting the outcomes generated by such institutions, starts with an outcome identified as desirable and attempts to construct a mechanism by which it may be delivered \cite{maskin2019introduction}. In this section we consider the competitive equilibrium concept discussed in prior sections as the target outcome for our market, and specify a mechanism by which it can be achieved. Mechanism design plays a crucial role for market in which participants may misreport preferences, costs or other information when it is in their individual best interest to do so. Therefore, the mechanism presented in this section may be viewed as a starting point for future work in which market participants are allowed to behave strategically.

The competitive equilibrium definition given in the previous section allows for non-binary activation schedule $\overline{x}^*$. As mentioned, since $0\leq \overline{x}^*_{it}\leq 1$, and $\sum_{t}\overline{x}^*_{it}\leq 1$, each $\overline{x}^*_{it}$ may be interpreted as giving the portion of load $i$ activated at time $t$ under relaxation of the binary constraints on the activation schedule or
the probability that an individual load of type $i$ in the infinite replication setting is fully activated at time $t$. 

Let us explore the infinitely replicated setting from the perspective of an individual load $n$ of type $i$. First, note that for finite $N$, (SPP($N$)-R) has Lagrangian \begin{equation}\nonumber\begin{split}
 &\mathcal{L} = \sum_tc\left(\hat{q}_t(N)\right)-\hat{\lambda}_t(N)(\hat{q}_t(N)+g_t)+\sum_{i,n,t}\left(\frac{u^{dS}_{it}}{N}(1-\hat{y}_{int}(N)) + \frac{u^{dE}_{it}}{N}(1-\hat{z}_{int}(N))\right)\\
 &\hspace{1.03in}- \sum_{i,n,t}\hat{\nu}^S_{int}(N)(1-\hat{y}_{int}(N))- \sum_{i,n,t}\tau_i\hat{\nu}^S_{int}(N)(1-\hat{z}_{int}(N))\\
&\hspace{1.03in}+\sum_{i,n,t}\left(\frac{p^{\hat{\lambda}}_{int}(N)}{N} + p^{\hat{\nu}}_{int}(N)\right)\hat{x}_{int}(N)-\sum_i\frac{\overline{U}_{i}}{N}\sum_n\sum_{t=1}^{T-\tau_i+1}\hat{x}_{int}(N)\\
\end{split}\end{equation}
where $p^{\hat{\lambda}}_{int}(N)$ and $p^{\hat{\nu}}_{int}(N)$ are defined analogously to (\ref{p_nu_lambda_def}) and (\ref{p_nu_lambda_def1}).

Thus, under $N$ replication and relaxation, the optimization problem for consumer $n$ of type $i$ is given by
\begin{align}\nonumber(\text{CONS}_{in}(N)\text{-R})\quad\min_{\substack{x_{in}^C,y^C_{in}\\\,z^C_{in}\geq 0}}&\quad \sum_t p^{\text{con}}_{int}(N)x^C_{int}-\frac{\overline{U}_{i}}{N}\sum_{t=1}^{T-\tau_i+1}x^C_{int}+\sum_t\left(\frac{u^{dS}_{it}}{N}-p^S_{int}\right)(1-y^C_{int})\\
\nonumber&\hspace{2.105in} +\sum_t\left(\frac{u^{dE}_{it}}{N}-p^E_{int}\right)(1-z^C_{int})
\end{align}
\begin{align}
\label{CON_in_const1}\text{s.t.}&\quad \theta^S_{int}\,:\,\sum_{s=1}^t\sum_{r=\max\{1,s-\tau_i+1\}}^sx^C_{inr}=\tau_i(1-y^C_{int})\quad\forall\,t\\
\label{CON_in_const2}&\quad \theta^E_{int}\,:\,\sum_{s=t}^T\sum_{r=\max\{1,s-\tau_i+1\}}^sx^C_{inr}=\tau_i(1-z^C_{int})\quad\forall\,t.
\end{align}
Multiplying by $N$, the $(\text{CON}_{in}(N)\text{-R})$ objective function can be written as
\begin{equation}\label{obj_times_N}\begin{split}
&\sum_tNp^{\text{con}}_{int}x^C_{int}+\sum_t\left(u^{dS}_{it}-Np^S_{int}\right)(1-y^C_{int})+\sum_t\left(u^{dE}_{it}-Np^E_{int}\right)(1-z^C_{int})-\overline{U}_{i}\sum_{t=1}^{T-\tau_i+1}x^C_{int}.
\end{split}\end{equation}
As in Theorem \ref{comp_eq_exists}, set
\begin{equation}\nonumber\begin{split}p^{\text{con}}_{int}(N) = \frac{p^{\hat{\lambda}*}_{int}(N)}{N} + p^{\hat{\nu}*}_{int}(N),\quad p^{\text{gen}*}_t(N)=\hat{\lambda}_t^*(N)\\
p^S_{int}(N)=\tau_i\hat{\nu}^{S*}_{int}(N),\quad p^E_{int}(N)=\tau_i\hat{\nu}^{E*}_{int}(N),
\end{split}\end{equation} 
and as in Proposition \ref{Nsol}, choose 
\begin{equation}\nonumber\begin{split}&\hspace{0.4in}\hat{\lambda}_t^*(N) = \hat{\lambda}_t^*(1)=\hat{\lambda}_t^*,\quad\hat{\nu}^{S*}_{int}(N)=\hat{\nu}^{S*}_{it}/N,\quad \hat{\nu}^{E*}_{int}(N)=\hat{\nu}^{E*}_{it}/N.\end{split}\end{equation}
Then letting $N\to\infty$ gives 
\begin{equation*}\begin{split}\lim_{N\to\infty}N\tau_i\hat{\nu}^{S*}_{int}(N) = \tau_i\hat{\nu}^{S*}_{it},\quad \lim_{N\to\infty}N\tau_i\hat{\nu}^{E*}_{int}(N) = \tau_i\hat{\nu}^{E*}_{it}.\end{split}\end{equation*}
This implies that 
\begin{equation}\nonumber\begin{split} \lim_{N\to\infty}Np^{\text{con}}_{int}(N) &= p^{\hat{\lambda}*}_{it} + p^{\hat{\nu}*}_{it}.
\end{split}\end{equation}
Therefore, posing the prices described in Proposition \ref{Nsol} in the limit as $N\to\infty$, the objective functions for each $(\text{CON}_{in})$ converge to 
\begin{equation}\nonumber\begin{split}
&\sum_t\left(p^{\hat{\lambda}^*}_{it}+ p^{\hat{\nu}^*}_{it}\right)x^C_{int}-\overline{U}_{i}\sum_{t=1}^{T-\tau_i+1}x^C_{int}+\sum_t\big(\left(u^{dS}_{it}-\tau_i\hat{\nu}^{S*}_{it}\right)(1-y^C_{int})+\left(u^{dE}_{it}-\tau_i\hat{\nu}^{E*}_{it}\right)(1-z^C_{int})\big).
\end{split}\end{equation}

Thus, the pricing facing each load of type $i$ is identical, and in fact the problem facing each is the same as the single load of type $i$ in the decomposition with relaxation but not replication. Further, each will select the same $x^{C*}_{in\cdot}=\hat{x}^*_{i\cdot}\in\mathbb{R}^{T}_+$, where $x^{C*}_{int}=\hat{x}^*_{it}$ gives the probability that the load will be scheduled at time $t$. Therefore the equal allocation for individuals of the same type mentioned earlier holds in our setting.

The following mechanism (\texttt{FLEX-SCHED}($N$)) uses the probability values selected by the continuum of consumers to generate a binary constrained schedule in the setting with $N$ replication. Note that since the generator's problem does not involve consumer utility and disutility functions, nor consumer scheduling variables, its problem is not affected by replication (or relaxation). Therefore $(\text{GEN}(N))$ is the same as $(\text{GEN})$ for all $N$, including $N=\infty$ . 
\begin{enumerate}
\item Each consumer $(i,n)$ submits $u^{dS}_{i\cdot}$ and $u^{dE}_{i\cdot}$, and the generator submits $c$ to the social planner (i.e. the entity taking on this role, such as the government or ISO).
\item The social planner solves (SPP-R), and announces $(\overline{p}^{\text{con}*},\overline{p}^{S*},\overline{p}^{E*},\overline{p}^{\text{gen}*},\overline{p}^{\text{bal}*})$ as specified in Theorem \ref{comp_eq_exists}.
\item Each consumer $i$ solves $(\text{CON($\infty$)-R}_{in})$, the generator solves $(\text{GEN}(\infty))$, and $(\overline{x}^*_i,\overline{y}^*_i,\overline{z}^*_i)$ for all $i$, as well as $\overline{q}^*$ are submitted to the social planner. 
\item The social planner randomly assigns proportion $\overline{x}^*_i$ of loads of type $i$ to start at time $t$, for each $i$ and $t$. The generator produces $\overline{q}^*$ over the finite horizon. Combined with the renewable generation output $g$, this generated power is allocated to the consumers according to $\overline{x}_i^*$ and demands $l_i$ for each $i$. 
\end{enumerate}

In the large population context, \emph{individual rationality} is achieved if the optimal objective values to $\text{GEN}(N)$, $\text{ISO}(N)$ and $(\text{CONS}_{in}(N))$ for all $i$ and $n$ with $N=\infty$, i.e., the individual entity problems under infinite replication but without relaxation, are positive under the  (\texttt{FLEX-SCHED}($\infty$)) solutions are nonnegative. \emph{Budget balance} is achieved if 
\begin{equation}\label{budget_bal_def}\begin{split}
\sum_tp^{\text{gen}}_t(q^G_t+g_t) + \sum_t\sum_i\sum_np^{S}_{int}(1-y^C_{int})+ \sum_t\sum_i\sum_np^{E}_{int}(1-z^C_{int})=\sum_i\sum_t\sum_np^{\text{con}}_{int}x^C_{int}.
\end{split}\end{equation}
Given these definitions, the following result regarding (\texttt{FLEX-SCHED}($\infty$)) holds. 

\begin{theorem}
The mechanism (\texttt{FLEX-SCHED}($\infty$)) is ex-post individually rational, budget balanced and efficient.
\end{theorem}
\begin{proof}
See Appendix \ref{thm_5_appendix}.
\end{proof}

\section{Multi-bus Extension}

Thus far, we have set aside the network model in order to focus on the scheduling and pricing aspects of this problem. As posed above, our formulation above could be interpreted as modeling the cases where either a single network point (e.g., EV charging station or smart building) is equipped with renewable generation, such as a solar panel array, as well as a backup generator \cite{forbessolarEV}. More broadly, our single bus model could represent the operation of a microgrid, with abstracted, aggregated renewable and backup generation, while neglecting distribution network features in serving the flexible loads. 


While we leave generalization to a full distribution network setting to future work, as a first step in this direction, we here consider an extension to a two bus model, in which the renewable generation and flexible loads are co-located at one bus, and the conventional generator is located at the other. 

We designate the bus where the loads and renewable generation are located as bus 1, and the bus where the conventional generator is located as bus 2. Selecting bus 1 as the slack bus, i.e., choosing from solutions with $\hat{\theta}_{1,t} = 0$ for all $t$ gives the following relaxation of the two-bus, centralized social planner's problem:
\begin{align}
\nonumber\min_{\substack{\hat{q},\hat{x},\hat{y},\hat{z},\hat{\theta}\geq 0}}&\quad\sum_tc(\hat{q}_t)+ \sum_i\sum_tu^{dS}_{it}(1-\hat{y}_{it})+\sum_i\sum_tu^{dE}_{it}(1-\hat{z}_{it})- \sum_i\overline{U}_i\sum_{t=1}^{T-\tau_i+1}\hat{x}_{it}\\
\text{s.t.}&\quad \hat{\lambda}_t\,:\, g_t- \sum_il_i\sum_{s=\max\{1,t-\tau_i+1\}}^t\hat{x}_{is} \geq -B\hat{\theta}_{2,t}\quad\forall\,t\\
\label{pf_12}\quad &\hat{\mu}_t\,:\, \hat{q}_t = B\hat{\theta}_{2,t}\quad\forall\,t\\
&\quad\nonumber\hat{\nu}^S_{it}\,:\,\sum_{s=1}^{t}\sum_{r=\max\{1,s-\tau_i+1\}}^s\hat{x}_{ir}\leq \tau_i(1-\hat{y}_{it})\quad\forall\,i,\,t\\
&\quad \nonumber\hat{\nu}^E_{it}\,:\,\sum_{s=t}^T\sum_{r=\max\{1,s-\tau_i+1\}}^s\hat{x}_{ir} \leq\tau_i(1-\hat{z}_{it})\quad\forall\,i,\,t\\
\label{flow_limit_12}&\quad  \hat{\gamma}_{12,t}\,:\,-B\hat{\theta}_{2,t}\leq f^{\max}\quad\forall\,t\\
\label{flow_limit_21}&\quad  \hat{\gamma}_{21,t}\,:\,B\hat{\theta}_{2,t}\leq f^{\max}\quad\forall\,t
\end{align}
Note that given the slack bus choice, $0\leq \hat{q}_t = B\hat{\theta}_{2,t}$, so that $-B\hat{\theta}_{2,t}<0$ for all $t$, i.e., constraint (\ref{flow_limit_12}) is always loose and $\hat{\gamma}^*_{12,t}=0$. 
See Appendix \ref{twobus_appendix} for the optimality conditions for this problem. 

Employing the traditional locational marginal pricing scheme yields the following nodal prices:
$$\overline{p}^*_{1,t} = \hat{\lambda}^*_t,\quad \overline{p}^*_{2,t} = \hat{\mu}^*_t.$$ 
It can be shown that at optimality, we have 
\begin{equation}\label{lambda_mu_ineq}\hat{\mu}^*_t = \hat{\lambda}^*_t - \hat{\gamma}^*_{21,t}\leq \hat{\lambda}^*_t, \end{equation}
so that the per unit energy price at bus 2, i.e., the rate paid to the conventional generator, is upper bounded by the per unit energy price paid by the collection of flexible loads. 

The two bus network model affects the individual entity problems only in terms of the generator and ISO objectives, as given below: 
\begin{equation}
\nonumber\text{(GEN-R)}\quad \max_{q^G\geq 0}\quad \sum_tp_{2,t}q^G_t - c(q^G_t)
\end{equation}
\begin{align*}
\text{(ISO-R)}\,\, \max_{\substack{q^I\geq 0\\x^I\geq 0,\,\theta^I}}&\quad\sum_tp_{1,t}\left(\sum_il_i\sum_{s=\max\{1,t-\tau_i+1\}}^tx^I_{is}-g_t\right) - \sum_tp_{2,t}q^I_t\\
\text{s.t.}&\quad \alpha_t\,:\,g_t - \sum_il_i\sum_{s=\max\{1,t-\tau_i+1\}}^tx^I_{is}\geq-B\theta^I_{2,t}\quad\forall\,t\\
&\quad \beta_t\,:\,q^I_t = B\theta^I_{2,t}\quad\forall\,t
\end{align*}
\begin{align*}
&\quad \xi_{12,t}\,:\,-B\theta^I_{2,t}\leq f^{\max},\, \xi_{21,t}\,:\,B\theta^I_{2,t}\leq f^{\max}\quad\forall\,t
\end{align*}
The consumer problem formulation does not change, as the energy consumption component of $p^{\text{con}}_{it}$ now includes a $p^{\hat{\lambda}^*}$ term derived from the nodal price at bus 1, rather than the dual variables to the single power balance constraint from our initial formulation. See Appendix \ref{twobus_appendix} for the aggregated individual entity problem optimality conditions.

The definition and proof of existence of a competitive equilibrium does not change in the two bus model, save for the addition of nodal prices, rather than the single per time slot energy consumption/production price taken as the dual variable to the power balance constraints over each time step. Nor do the proofs of the relationship between competitive equilibria from the original relaxed setting, and the relaxed settings featuring replication. 

Given the shift from a single per time slot energy consumption/production price to a pair of prices, one for consumption at node 1, and one for production at node 2, the only remaining question is whether Theorem 5 continues to hold, in particular the budget balance component. Essentially due to inequality (\ref{lambda_mu_ineq}), instead of budget balance, we now have budget adequacy, i.e., under the proposed scheduling and pricing schemes, our market mechanism either is either budget balanced, or incurs a surplus in the form of congestion revenue.

To see this, the inequality which we would like to show is the following : 
\begin{equation}\nonumber
\begin{split}
&\sum_tp^{\text{gen}}_tq^G_t + \sum_t\sum_i\sum_n\left(p^S_{int}(1-y^C_{int}) +p^E_{int}(1-z^C_{int})\right) \geq  \sum_i\sum_t\sum_n p^{\text{con}}_{int}x^C_{int}= \sum_i\sum_t\left(p^{\hat{\lambda}*}_{it} + p^{\hat{\nu}*}_{it}\right)\hat{x}^*_{it}.
\end{split}
\end{equation}
Note that if for a time slot $t$
\begin{equation}\nonumber g_t - \sum_il_i\sum_{s=\max\{1,t-\tau_i+1\}}^t \hat{x}^*_{is}\geq 0,\end{equation}
then we have that $\hat{\theta}^*_{2,t}=\hat{q}^*_t=\hat{\lambda}^*_t=\hat{\mu}^*_t=\hat{\gamma}^*_{21,t}=0$,
i.e., the generator is not dispatched at time $t$. 
For all other time slots, due to the power flow constraints, we have that 
\begin{equation}\label{budget_bal_quant}\hat{q}^*_t = \sum_il_i\sum_{s=\max\{1,t-\tau_i+1\}}^t\hat{x}^*_{is} - g_t.\end{equation}
However, again due to the two-node relaxed social planner's problem KKT conditions, 
\begin{equation}\label{budget_bal_price}\hat{\mu}^*_t = \hat{\lambda}^*_t - \hat{\gamma}^*_{21,t}\leq \hat{\lambda}^*_t\quad\forall\,t.\end{equation}
Together, (\ref{budget_bal_quant}) and (\ref{budget_bal_price}) show that in problem instances where the line connecting the generator and load nodes becomes congested at any time slot, the summed load payments will in general exceed payments to the generator : 
$$\sum_t\hat{\mu}^*_t\hat{q}^*_t\leq \hat{\lambda}^*_t\left(\sum_il_i\sum_{s=\max\{1,t-\tau_i+1\}}^t\hat{x}^*_{is} - g_t\right).$$
In such cases the system operator will collect nonzero congestion revenue. Alternatively, this surplus could be interpreted as a penalty assigned to loads operating at the time of congestion. Therefore, the mechanism is either budget balanced or runs a surplus for the system operator. The congestion revenue may be handled via a separate market, e.g. a financial transmission rights market.

\section{Case Study: EV Charging}

Electric vehicle charging constitutes one of the most important and challenging applications of load scheduling optimization currently facing power grid operators. Today, the transportation sector accounts for approximately 64\% of global consumption of oil, a resource which has been linked to increasing CO2 emissions, and further is expected to expire in about 50 years. In contrast, transportation sector operations comprise just 1.5\% of worldwide electricity usage \cite{elreview}. Reliance on electricity is more amenable to a shift towards renewable sources of energy such as solar and wind, which in total are expected to make up approximately one third of all power generation by 2040. From a market perspective, demand for electric vehicles increase each year. According to the International Energy Agency, 740,000 vehicles were produced in total in 2014, and that figure is expected to reach 20 million by 2020 \cite{elreview}. Charging process scheduling is now recognized as one of the key technologies for integration of electricity based mobility into existing power grids.

In order to demonstrate the utility of our flexible scheduling problem formulation, we simulated its performance on real world load and renewable generation data in the context of electric vehicle (EV) charging. The input load parameters $(\tau_i,l_i,u^{dS}_i,u^{dE}_i)$ are derived from data included in the ACN-Data dataset, a dynamic dataset of workplace EV charging \cite{lee2019acn}. In particular, we take as our base set of loads the recorded vehicle arrivals for May 28, 2018. For each vehicle charging session, the ACN dataset includes vehicle connection and disconnection times, as well as a charging completion time. In these simulations each time index represents a 15 minute period. We take $\tau_i$ as the difference between charging completion time and the first time period where the EV drew a positive amount of current. We then divide the total kWh delivered to the vehicle by $\tau_i$ to arrive at $l_i$. 

We design disutility functions for each load in the following manner. Let $t_C$ denote the vehicle's recorded connection time, and $t_D$ denote its recorded disconnection time. Then for each load $i$, we let 
\begin{equation}\label{orig_disu}\begin{split}
u^{dS}_{it} &= \begin{cases}
\alpha(t_C-t)^2&0\leq t\leq t_C-1\\
0&t\geq t_C
\end{cases}\quad\text{and}\quad
u^{dE}_{it} = \begin{cases}
0&t\leq t_D\\
\alpha(t-T_E)^2&t_E+1\leq t\leq T
\end{cases},
\end{split}
\end{equation}
where $\alpha$ is a scaling parameter.
Thus $u^{dS}_{it} + u^{dE}_{it}=0$ for $t\in[t_C,t_D]$, and elsewhere increases quadratically away from this desired service window. The sample disutility curves pictured in Figure 
\ref{StartEndDisutility} are generated according to this quadratic form. 

 We set $\overline{U}_i= \overline{U}$ for some nonnegative scalar $\overline{U}$. We take our renewable generation profile $g_t$ from data generated by NREL's SAM tool \cite{blair2014system}. Specifically, we draw on solar power generation time series estimated for downtown Los Angeles, also for May 28, 2018. 

In our simulations we compare the performance of our flexible load scheduling approach to a schedule which naively begins charging loads as soon as they arrive. We implemented the latter approach by setting
\begin{equation}\nonumber\begin{split}
u^{dS'}_{it} = \begin{cases}
\alpha\max\{t_C^2,(T-t_C)^2\}&0\leq t\leq t_C-1\\
0&t\geq t_C\\
\end{cases}\quad\text{and}\quad
u^{dE'}_{it} = \begin{cases}
0&t\leq t_C\\
\alpha\max\{t_C^2,(T-t_C)^2\}&t_C+1\leq t\leq T\\
\end{cases},
\end{split}
\end{equation}

i.e., loads are essentially inflexible, with 0 disutility at time $t_C$ and the maximum disutility in (\ref{orig_disu}) for all other time periods $t\neq t_C$.
We consider the social welfare objective value achieved, as well as the percentage of loads served. 

We adjusted $\alpha$, $\overline{U}$, and a scale on the quadratic cost $c(z_t)$ in order to ensure that both scheduling approaches successfully scheduled all loads. In particular we let $\alpha=0.01$, $\overline{U}=100$ and scaled the cost by factor $0.5$. As shown in Figures \ref{fig:AggLoads} and \ref{fig:ThermalGeneration}, when users report disutility functions, the scheduler shifts loads such that that overall demand moves towards the mid day period of high renewable generation, thus relying less on the generator and therefore incurring less cost and higher social welfare overall. For the base load case examined in Figure 4, the peak generation falls from 4.67 kW to 3.46 kW, a reduction of roughly 29\%, while the peak demand falls from 7.23 kW to 5.46 kW, a reduction of roughly 24\%. 

\begin{figure}[h]
\begin{center}
\includegraphics[width=3.5in]{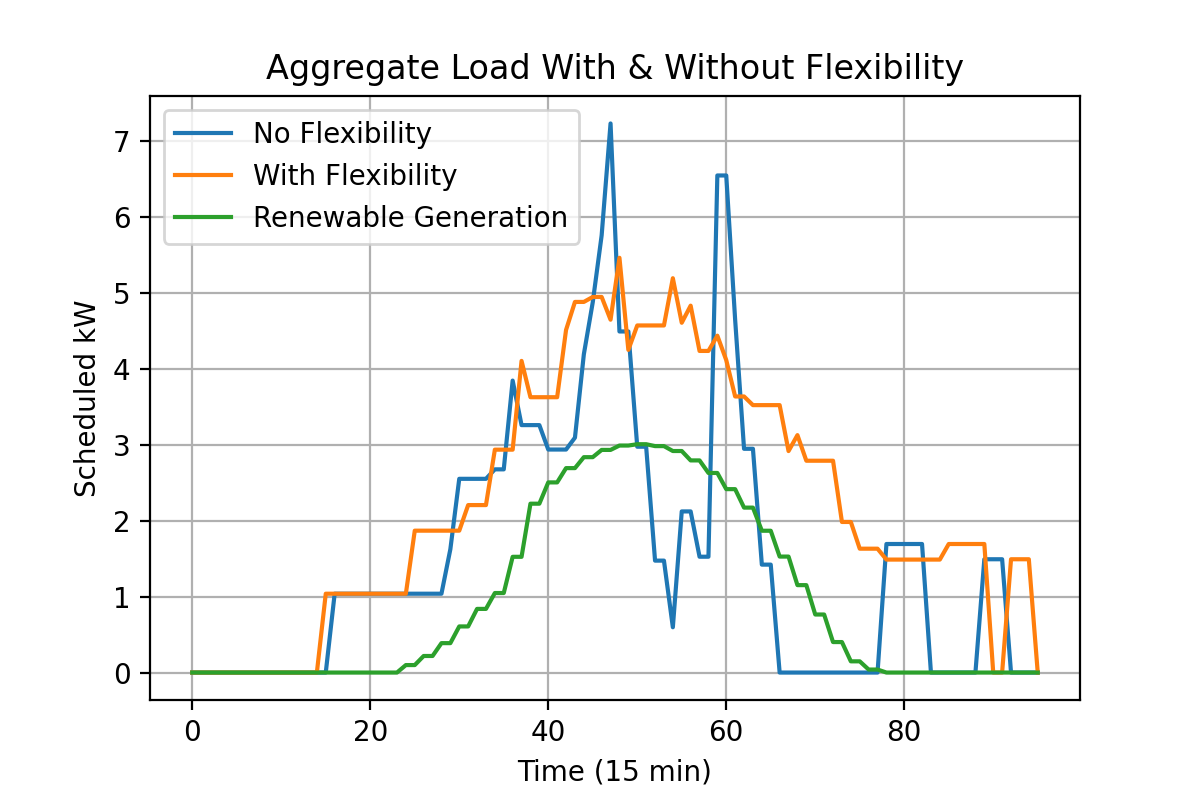}
\caption{Scheduled aggregate load with and without flexibility.}
\label{fig:AggLoads}
\end{center}
\end{figure}

\begin{figure}[h]
\begin{center}
\includegraphics[width=3.5in]{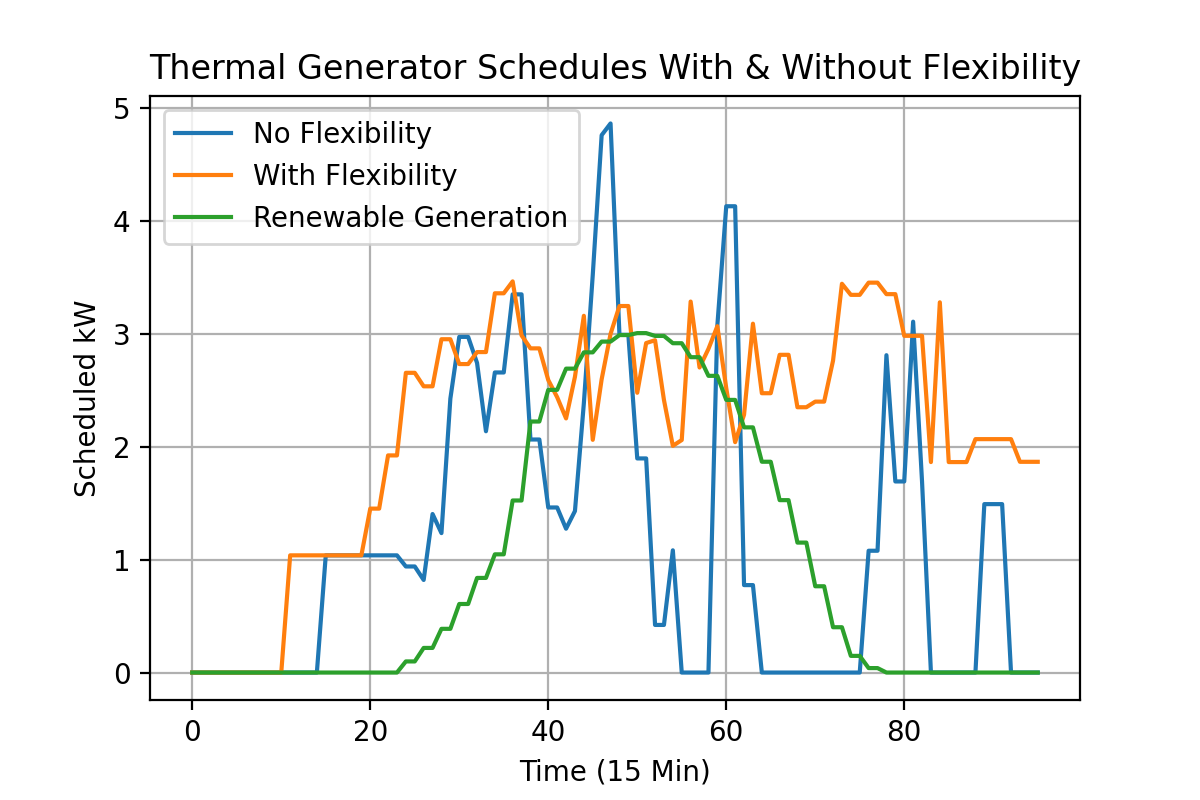}
\caption{Scheduled generation (equal to aggregate load less renewable generation) with and without flexibility.}
\label{fig:ThermalGeneration}
\end{center}
\end{figure}

To demonstrate the robustness of the disutility function approach to a surge in demand, we randomly sample loads served during the other weekdays of May 2018 in order to increase overall power demand. Specifically, we study demand scaled up in increments of 25\% of the base load up to a 100\% increase in load in terms of overall demand. Loads are randomly added to the base until the total power demand exceeds the desired level of increase, and the same random sets of loads are added in the cases with and without flexibility. The performance of both approaches are shown in Figures \ref{fig:PercServed} and \ref{fig:SocialWelfare}. The flexibility enabled schedule which makes use of the disutility functions continues to offer increased social welfare over on demand scheduling. Additionally, while disutility based scheduling still includes all loads, the on demand based scheduler finds it optimal to exclude between 25\% and 33\% of loads as the number of loads increases to double the base.  

\begin{figure}[h]
\begin{center}
\includegraphics[width=3.5in]{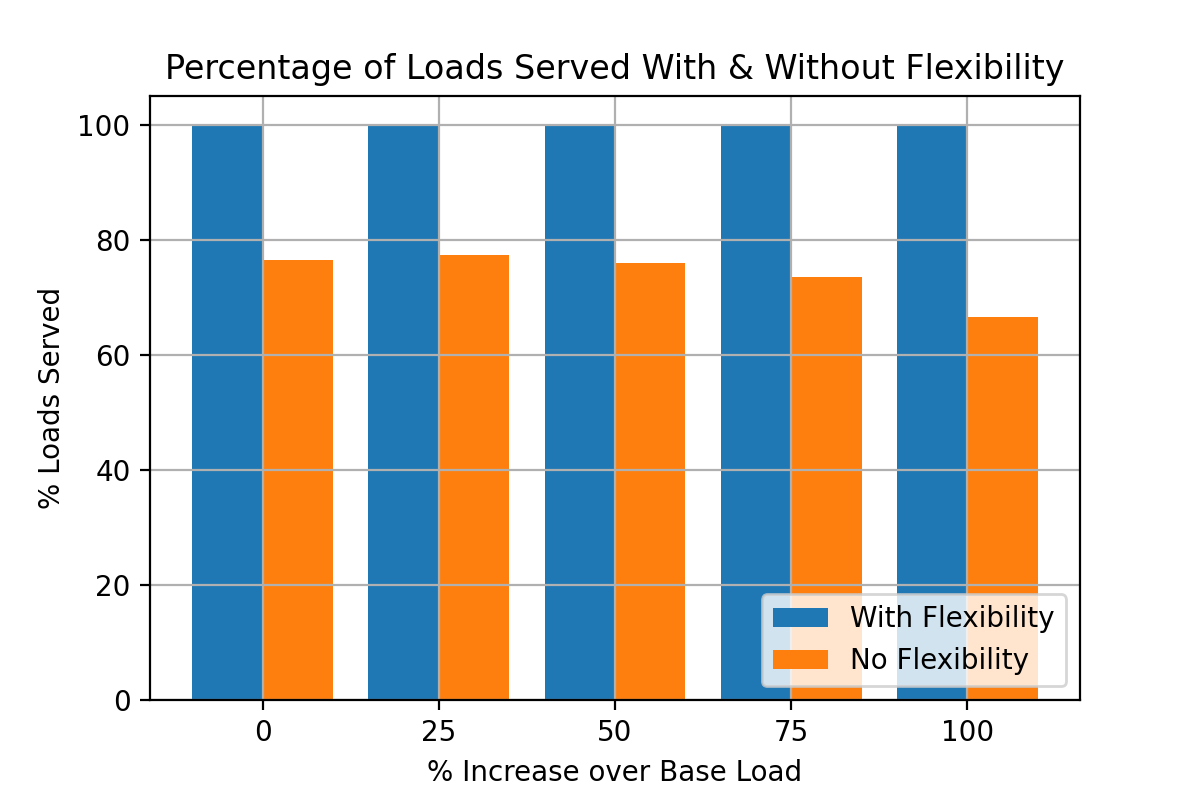}
\caption{Proportions of loads served with and without flexibility.}
\label{fig:PercServed}
\end{center}
\end{figure}

\begin{figure}[h]
\begin{center}
\includegraphics[width=3.5in]{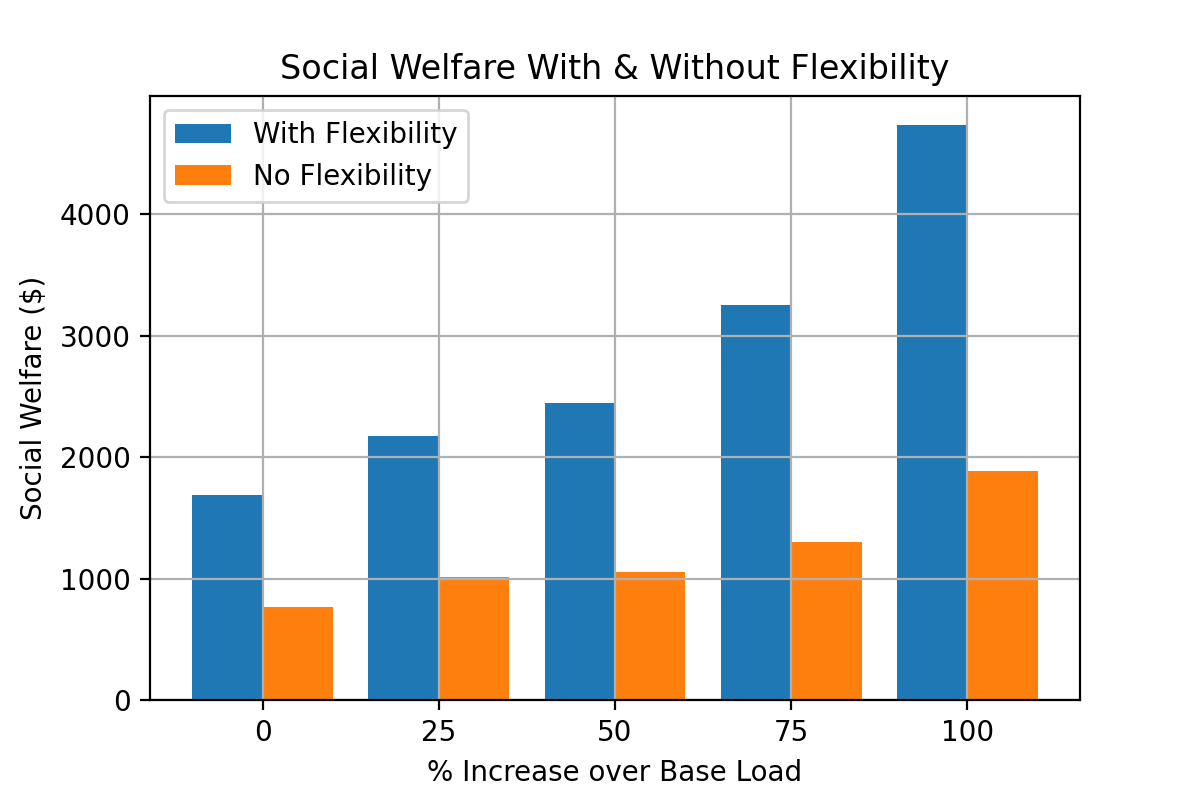}
\caption{Social welfare achieved with and without flexibility.}
\label{fig:SocialWelfare}
\end{center}
\end{figure}

\section{Conclusion}
In this work, we study how to schedule and price service for a population of flexible, but non-preemptive loads, in the presence of renewable generation, as well as a dispatchable thermal generator. Formulating a collection of mixed integer optimization programs for the consumers, and generator, we then study a centralized version of our setting with relaxed integer constraints, allowing for use of Lagrangian analysis and derivation of prices. A solution of this centralized problem yields a competitive equilibrium, and conversely a competitive equilibrium yields an efficient solution. Finally, we present a case study involving electric car charging data to demonstrate the efficacy of our approach. 

There are several directions for future work in this area. First, in terms of the scheduling aspect, it is desirable to determine a method for deriving at least an approximately optimal solution to the original integer constrained setting, given an efficient solution to the relaxed social planner's problem presented here. In terms of pricing, properties such as fairness should be examined. For example, assuming that the disutility functions of each user can be at least partially ordered from less to more restrictive, is the compensation offered to more flexible users more than to those which are not as flexible? It will also be of interest to explore other types of loads, such as those which may be interrupted, as well as those which might accept less than an upper bound of total energy delivered. Strategic behavior amongst market participants should also be taken into account, as well as more detailed network modeling and constraints. 

\bibliography{Bibliography}
\bibliographystyle{plain}

\appendices

\section{$p^{\hat{\lambda}}$ and $p^{\hat{\nu}}$ Derivation}\label{p_lambda_nu_deriv}
Start with $p^{\hat{\lambda}}$ and the double sum 
\begin{equation}\nonumber\begin{split}&\sum_{t=1}^T\hat{\lambda}_t\sum_{s=\max\{1,t-\tau_i+1\}}^t\hat{x}_{is}=\sum_{t=1}^{\tau_i}\hat{\lambda}_t \sum_{s=1}^{t}\hat{x}_{is}+ \sum_{t=\tau_i+1}^{T}\hat{\lambda}_t\sum_{s=t-\tau_i+1}^{t}\hat{x}_{is}.\end{split}\end{equation}
Rearranging the first double sum, and breaking the second double sum on the right into separate sums based upon whether the $t$ value appears in the upper and/or lower argument of the inner sum:
 \begin{equation}\nonumber\begin{split}
&=\sum_{t=1}^{\tau_i}\hat{x}_{it}\sum_{s=t}^{\tau_i}\hat{\lambda}_s+\sum_{t=2}^{\tau_i}\hat{x}_{it}\sum_{s=\tau_i+1}^{t+\tau_i-1}\hat{\lambda}_s+ \sum_{t=\tau_i+1}^{T-\tau_i+1}\hat{x}_{it}\sum_{s=t}^{t+\tau_i-1}\hat{\lambda}_s + \sum_{t=T-\tau_i+2}^T\hat{x}_{it}\sum_{s=t}^T\hat{\lambda}_s\\
&=\sum_{t=1}^{\tau_i}\hat{x}_{it}\sum_{s=t}^{t+\tau_i-1}\hat{\lambda}_s+ \sum_{t=\tau_i+1}^{T-\tau_i+1}\hat{x}_{it}\sum_{s=t}^{t+\tau_i-1}\hat{\lambda}_s + \sum_{t=T-\tau_i+2}^T\hat{x}_{it}\sum_{s=t}^T\hat{\lambda}_s\\
&=\sum_{t=1}^{T-\tau_i+1}\hat{x}_{it}\sum_{s=t}^{t+\tau_i-1}\hat{\lambda}_s + \sum_{t=T-\tau_i+2}^T\hat{x}_{it}\sum_{s=t}^T\hat{\lambda}_s\\
&=\sum_{t=1}^{T}\hat{x}_{it}\sum_{s=t}^{\min\{t+\tau_i-1,T\}}\hat{\lambda}_s
\end{split}\end{equation}
Therefore, we define
\begin{equation}\nonumber p^{\hat{\lambda}}_{it} := l_i\sum_{s=t}^{\min\{t+\tau_i-1,T\}}\hat{\lambda}_s.\end{equation}
Now, we can rewrite the double sum
\begin{equation}\nonumber\begin{split}\sum_{s=1}^t\sum_{r=\max\{1,s-\tau_i+1\}}^s\hat{x}_{ir}&=\sum_{s=1}^{t}\hat{x}_{is}\sum_{r=s}^{\min\{s+\tau_i-1,t\}}1=\sum_{s=1}^{t}\hat{x}_{is}\min\{\tau_i,t-s+1\}.
\end{split}
\end{equation}
Then, to develop the definition of $p^{\hat{\nu}}$, start with:
\begin{equation}\nonumber\begin{split}&\sum_{t=1}^T\hat{\nu}^S_{it}\sum_{s=1}^{t}\hat{x}_{is}\min\{\tau_i,t-s+1\}= \sum_{t=1}^T\sum_{s=1}^{t}\hat{\nu}^S_{it}\hat{x}_{is}\min\{\tau_i,t-s+1\}
= \sum_{t=1}^T\hat{x}_{it}\sum_{s=t}^{T}\hat{\nu}^S_{is}\min\{\tau_i,s-t+1\}.
\end{split}\end{equation}
Therefore, each $\hat{x}_{it}$ has the additional coefficient 
$$\sum_{s=t}^{T}\hat{\nu}^S_{is}\min\{\tau_i,s-t+1\}.$$
Finally, considering the double sum 
\begin{align*}\sum_{s=t}^T\sum_{r=\max\{1,s-\tau_i+1\}}^s\hat{x}_{ir}&= \sum_{s=1}^T\sum_{r=\max\{1,s-\tau_i+1\}}^s\hat{x}_{ir}-\sum_{s=1}^{t-1}\sum_{r=\max\{1,s-\tau_i+1\}}^s\hat{x}_{ir}\\
&= \sum_{s=1}^{T}\hat{x}_{is}\min\{\tau_i,T-s+1\} - \sum_{s=1}^{t-1}\hat{x}_{is}\min\{\tau_i,t-s\} \\
&= \sum_{s=1}^{T-\tau_i+1}\hat{x}_{is}\tau_i +  \sum_{s=T-\tau_i+2}^{T}\hat{x}_{is}(T-s+1) - \sum_{s=1}^{t-\tau_i}\hat{x}_{is}\tau_i- \sum_{s=\max\{1,t-\tau_i+1\}}^t\hat{x}_{is}(t-s)\\
&= \sum_{s=\max\{1,t-\tau_i+1\}}^{T-\tau_i+1}\hat{x}_{is}\tau_i +  \sum_{s=T-\tau_i+2}^{T}\hat{x}_{is}(T-s+1)- \sum_{s=\max\{1,t-\tau_i+1\}}^t\hat{x}_{is}(t-s)\\
&= \sum_{s=\max\{1,t-\tau_i+1\}}^{t}\hat{x}_{is}(\tau_i-(t-s)) +  \sum_{s=t+1}^{T-\tau_i+1}\hat{x}_{is}\tau_i +  \sum_{s=T-\tau_i+2}^{T}\hat{x}_{is}(T-s+1)\\
&= \sum_{s=\max\{1,t-\tau_i+1\}}^{T}\hat{x}_{is}\min\{\tau_i-(t-s),\tau_i,T-s+1\}
\end{align*}
Then 
\begin{equation}\nonumber\begin{split}
\sum_{t=1}^T\hat{\nu}^E_{it}\sum_{s=\max\{1,t-\tau_i+1\}}^{T}\hat{x}_{is}\min\{\tau_i-(t-s),\tau_i,T-s+1\}&= \sum_{t=1}^T\sum_{s=\max\{1,t-\tau_i+1\}}^{T}\hat{\nu}^E_{it}\hat{x}_{is}\min\{\tau_i-(t-s),\tau_i,T-s+1\} \\
  &= \sum_{t=1}^T\hat{x}_{it}\sum_{s=1}^{\min\{t+\tau_i-1,T\}}\hat{\nu}^E_{is}\min\{\tau_i-(s-t),\tau_i,T-t+1\}.
\end{split}\end{equation}
Therefore, the coefficient for each $\hat{x}_{it}$ contains the term 
$$ \sum_{s=1}^{\min\{t+\tau_i-1,T\}}\hat{\nu}^E_{is}\min\{\tau_i-(s-t),\tau_i,T-t+1\},$$
and thus we define: 
\begin{equation*}\begin{split}&p^{\hat{\nu}}_{it}:=\sum_{s=t}^{T}\hat{\nu}^S_{is}\min\{s-t+1,\tau_i\}+ \sum_{s=1}^{\min\{t+\tau_i-1,T\}}\hat{\nu}^E_{is}\min\{T-t+1,\tau_i,\tau_i-(s-t),\}.\end{split}\end{equation*}

\section{Individual Entity Relaxed Problem Optimality Conditions}\label{entity_optimality}
The optimality conditions for $(\text{CONS-R}_i)$ are 
\begin{align*}
p^{\text{con}}_{it} - \overline{U}_i + p^{\theta^*}_{it}&\geq 0\quad\forall\,i,\,t\leq T-\tau_i+1\\
x^{C*}_{it}\left(p^{\text{con}}_{it} - \overline{U}_i + p^{\theta^*}_{it}\right)&= 0\quad\forall\,i,\,t\leq T-\tau_i+1
\end{align*}
\begin{align*}
p^{\text{con}}_{it} + p^{\theta^*}_{it}&\geq 0\quad\forall\,i,\,t> T-\tau_i+1\\
x^{C*}_{it}\left(p^{\text{con}}_{it}+ p^{\theta^*}_{it}\right)&= 0\quad\forall\,i,\,t> T-\tau_i+1\\
p^S_{it}- u^{dS}_{it}+\tau_i\theta^{S*}_{it}&\geq 0\quad\forall\,i,\,t\\
y^{C*}_{it}\left(p^S_{it}- u^{dS}_{it}+\tau_i\theta^{S*}_{it}\right)&=0\quad\forall\,i,\,t\\
p^E_{it}-u^{dE}_{it}+\tau_i\theta^{E*}_{it}&\geq 0\quad\forall\,i,\,t\\
z^{C*}_{it}\left(p^E_{it}-u^{dE}_{it}+\tau_i\theta^{E*}_{it}\right)&=0\quad\forall\,i,\,t,
\end{align*}
where $p^{\theta^*}_{it}$ is defined analogously to $p^{\hat{\nu}}_{it}$ in (\ref{p_nu_lambda_def}). The optimality conditions for (GEN-R) are
\begin{align*}
c'(q^{G*}_t) - p^{\text{gen}}_t&\geq 0\quad\forall\,t\\
q^{G*}_t\left(c'(q^{G*}_t) - p^{\text{gen}}_t\right)&= 0\quad\forall\,t.
\end{align*}
The optimality conditions for (ISO-R) are
\begin{align}
\label{ISO-RKKT1}p^{\text{gen}}_t-\alpha^*_t&\geq 0\quad\forall\,t\\
\label{ISO-RKKT2}q^{I*}_t\left(p^{\text{gen}}_t-\alpha^*_t\right)&=0\quad\forall\,t\\
\label{ISO-RKKT3}-p^{p^{\text{gen}}}_{it} + p^{\alpha^*}_{it}&\geq 0 \quad\forall\,t\\
\label{ISO-RKKT4}x^{I*}_{t}\left(-p^{p^{\text{gen}}}_{it} + p^{\alpha^*}_{it}\right)&= 0 \quad\forall\,t,\\
\label{ISO-RKKT5}\alpha^*_t\left(\sum_il_i\sum_{s=\max\{1,t-\tau_i+1\}}^tx^{I*}_{is}-g_t- q^{I*}_t\right)&=0\quad\forall\,t\\
\label{ISO-RKKT6}\alpha^*_t&\geq 0\quad\forall\,t,
\end{align}
where $p^{p^{\text{gen}}}_{it} $ and $p^{\alpha^*}_{it}$ are defined analogously to $p^{\hat{\lambda}}_{it}$ in (\ref{p_nu_lambda_def}). Note that (\ref{ISO-RKKT1})-(\ref{ISO-RKKT4}) can always be satisfied by choosing $\alpha^*_t=p^{\text{gen}}_t$ for all $t$. Therefore, only constraints (\ref{ISO-RKKT5}) and (\ref{ISO-RKKT6}), along with feasibility need be considered assuming $\alpha^*_t=p^{\text{gen}}_t$ for all $t$. 

\section{Proof of Theorem 1}\label{thm_1_appendix}
\begin{proof}
  Given price selections according to (\ref{priceselect}), and selecting $q^{G}=q^{I} = \hat{q}^*$, $x^C_i=x^I_i=\hat{x}^*_i$ for all $i$, $y^C_i=\hat{y}_i^*$ and  $z_i^C=\hat{z}_i^*$ for all $i$ makes the collected optimality conditions (aside from feasibility) for $(\text{CON-R}_i)$ for each $i$, (GEN-R) and (ISO-R)
  \begin{align}
c'(\hat{q}^{*}_t) - \hat{\lambda}^*_t&\geq 0\quad\forall\,t\\
\hat{q}^{*}_t\left(c'(\hat{q}^{*}_t) - \hat{\lambda}^*_t\right)&= 0\quad\forall\,t\\
p^{\hat{\lambda}*}_{it}  + p^{\hat{\nu}*}_{it} - \overline{U}_i + p^{\theta*}_{it}&\geq 0\quad\forall\,i,\,t\leq T-\tau_i+1\\
\hat{x}^*_{it}\left(p^{\hat{\lambda}*}_{it}  + p^{\hat{\nu}*}_{it} - \overline{U}_i + p^{\theta*}_{it}\right)&= 0\quad\forall\,i,\,t\leq T-\tau_i+1
\end{align}
\begin{align}
p^{\hat{\lambda}*}_{it}  + p^{\hat{\nu}*}_{it} + p^{\theta*}_{it}&\geq 0\quad\forall\,i,\,t> T-\tau_i+1\\
\hat{x}^*_{it}\left(p^{\hat{\lambda}*}_{it}  + p^{\hat{\nu}*}_{it}+ p^{\theta*}_{it}\right)&= 0\quad\forall\,i,\,t> T-\tau_i+1\\
\hat{\nu}^{S*}_{it}\tau_i-u^{dS}_{it}+\theta^{S*}_{it}\tau_i&\geq 0\quad\forall\,i,\,t\\
\hat{y}^*_{it}\left(\hat{\nu}^{S*}_{it}\tau_i-u^{dS}_{it}+\theta^{S*}_{it}\tau_i\right)&=0\quad\forall\,i,\,t\\
\hat{\nu}^{E*}_{it}\tau_i-u^{dE}_{it}+\theta^{E*}_{it}\tau_i&\geq 0\quad\forall\,i,\,t\\
\hat{z}^{*}_{it}\left(\hat{\nu}^{E*}_{it}\tau_i-u^{dE}_{it}+\theta^{E*}_{it}\tau_i\right)&=0\quad\forall\,i,\,t\\
\hat{\lambda}^*_t\left(\sum_il_i\sum_{s=\max\{1,t-\tau_i+1\}}^tx^{I*}_{is}-g_t- q^{I*}_t\right)&=0\quad\forall\,t\\
\hat{\lambda}^*_t\geq 0&\quad\forall\,t.
  \end{align}

Further selecting $\theta^{S*}_{it} =\theta^{E*}_{it}=0$ for all $i$ and $t$ gives
  \begin{align}
 c'(\hat{q}^{*}_t) - \hat{\lambda}^*_t&\geq 0\quad\forall\,t\\
\hat{q}^{*}_t\left(c'(\hat{q}^{*}_t) - \hat{\lambda}^*_t\right)&= 0\quad\forall\,t\\
 p^{\hat{\lambda}*}_{it}  + p^{\hat{\nu}*}_{it} - \overline{U}_i&\geq 0\quad\forall\,i,\,t\leq T-\tau_i+1\\
\hat{x}^*_{it}\left(p^{\hat{\lambda}*}_{it}  + p^{\hat{\nu}*}_{it} - \overline{U}_i \right)&= 0\quad\forall\,i,\,t\leq T-\tau_i+1\\
p^{\hat{\lambda}*}_{it}  + p^{\hat{\nu}*}_{it} &\geq 0\quad\forall\,i,\,t> T-\tau_i+1\\
\hat{x}^*_{it}\left(p^{\hat{\lambda}*}_{it}  + p^{\hat{\nu}*}_{it}\right)&= 0\quad\forall\,i,\,t> T-\tau_i+1\\
\hat{\nu}^{S*}_{it}\tau_i -u^{dS}_{it}&\geq 0\quad\forall\,i,\,t\\
\hat{y}^*_{it}\left(\hat{\nu}^{S*}_{it}\tau_i -u^{dS}_{it}\right)&=0\quad\forall\,i,\,t\\
\hat{\nu}^{E*}_{it}\tau_i-u^{dE}_{it}&\geq 0\quad\forall\,i,\,t\\
\hat{z}^{*}_{it}\left(\hat{\nu}^{E*}_{it}\tau_i-u^{dE}_{it}\right)&=0\quad\forall\,i,\,t\\
\hat{\lambda}^*_t\left(\sum_il_i\sum_{s=\max\{1,t-\tau_i+1\}}^tx^{I*}_{is}-g_t- q^{I*}_t\right)&=0\quad\forall\,t\\
\hat{\lambda}^*_t&\geq 0\quad\forall\,t.
  \end{align}
  These expressions are identical to the (SPP-R) KKT conditions, and therefore satisfied by optimal solutions to (SPP-R). As a primal solution to (SPP-R), $(\hat{x}^*,\hat{y}^*,\hat{z}^*)$ also satisfies the collected constraints from $(\text{CON-R}_i)$ for each $i$, (GEN-R) and (ISO-R). 
\end{proof}
\section{Proof of Theorem 5}\label{thm_5_appendix}\begin{proof}
Let us denote realizations of the randomized scheduled specified by $\overline{x}^*$ as $\tilde{x}$, and similarly for other variables. 

Starting with ex-post individual rationality, suppose that for some $i$ we have that $\sum_t\overline{x}^*_{it}<1$, so that a portion of population $i$ will not be activated. Then $\tilde{x}_{int} = 0$ for all $t$, and $\tilde{y}_{int}=\tilde{z}_{int}=1$ for all $t$, so that the objective of $(\text{CON($\infty$)-R}_{in})$, i.e., the realized net utility of load $n$ of type $i$ is equal to 0. 

In all other cases, load $n$ of type $i$ is scheduled, so that for some $\tilde{t}_{in}$ where (SPP-R) KKT conditions (\ref{SPPRKKT3}) and (\ref{SPPRKKT4}) are satisfied, $\tilde{x}_{int}=1$. Due to constraints (\ref{CON_in_const1}) and (\ref{CON_in_const2}), in $(\text{CON($\infty$)-R}_{in})$, we have that 
\begin{equation}\begin{split}
\tilde{y}_{int} = \begin{cases}
1&t\leq \tilde{t}_{in}-1\\
1-\frac{t-\tilde{t}_{in}+1}{\tau_{i}}&\tilde{t}_{in}\leq t\leq \tilde{t}_{in}+\tau_i-2\\
0&t\geq \tilde{t}_{in}+\tau_i-1
\end{cases}\\
\tilde{z}_{int} = \begin{cases}
0&t\leq \tilde{t}_{in}\\
\frac{t-\tilde{t}_{in}}{\tau_{i}}&\tilde{t}_{in}+1\leq t\leq \tilde{t}_{in}+\tau_i-1\\
1&t\geq \tilde{t}_{in}+\tau_i
\end{cases}.
\end{split}
\end{equation}
Substituting $\tilde{x}$, $\tilde{y}$ and $\tilde{z}$ and the equilibrium prices from step 2 of (\texttt{FLEX-SCHED}) into expression (\ref{obj_times_N}) gives 
\begin{equation}\label{sub_obj_times_N}\begin{split}p^{\hat{\lambda}*}_{i\tilde{t}_{in}} &+ p^{\hat{\nu}*}_{i\tilde{t}_{in}} - \overline{U}_i+\sum_{t=\tilde{t}_{in}+1}^{\tilde{t}_{in}+\tau_i-1}\frac{t-\tilde{t}_{in}}{\tau_{i}}(u^{dE}_{it} - \tau_i\hat{\nu}^{E*}_{it})\\
&+\sum_{t=\tilde{t}_{in}+\tau_i-1}^{T}(u^{dS}_{it} - \tau_i\hat{\nu}^{S*}_{it})+\sum_{t=1}^{\tilde{t}_{in}}(u^{dE}_{it} - \tau_i\hat{\nu}^{E*}_{it})\\
&+\sum_{t=\tilde{t}_{in}}^{\tilde{t}_{in}+\tau_i-2}\left(1-\frac{t-\tilde{t}_{in}+1}{\tau_{i}}\right)(u^{dS}_{it} - \tau_i\hat{\nu}^{S*}_{it})
 \end{split}\end{equation}
 The term $p^{\hat{\lambda}*}_{i\tilde{t}_{in}} + p^{\hat{\nu}*}_{i\tilde{t}_{in}} - \overline{U}_i$ is equal to 0 due to (SPP-R) KKT conditions (\ref{SPPRKKT3}) and (\ref{SPPRKKT4}) and the fact that $\tilde{t}$ is a time index where $\hat{x}_{i\tilde{t}}>0$. The terms in the sums are nonpositive due to (SPP-R) KKT conditions (\ref{SPPRKKT5}) and (\ref{SPPRKKT7}). Thus, each user will incur nonnegative net utility when participating in the mechanism, regardless of whether or not they are scheduled. 

Selecting $x^C_{int} = \overline{x}^*_{it}= \hat{x}^*_{it} = $ for all $n$, and 
\begin{equation}p^{\text{con}}_{int} =\overline{p}^{\text{con}*}_{it}= \frac{p^{\hat{\lambda}*}_{it}+p^{\hat{\nu}*}_{it}}{N},\end{equation}
the right hand side of (\ref{budget_bal_def}) is equal to 
\begin{equation}\label{budget_bal_proof1}\begin{split}
\sum_i\sum_tp^{\hat{\lambda}*}_{it}\hat{x}^*_{it}+\sum_i\sum_tp^{\hat{\nu}*}_{it}\hat{x}^*_{it}.
\end{split}\end{equation}
From the definition of $p^{\hat{\lambda}^*}_{it}$ and the power balance constraint (\ref{SPPconst1}) in (SPP-R), the left term in (\ref{budget_bal_proof1}) is equal to 
\begin{equation}\begin{split}\label{budget_bal_proof2}&\sum_t\hat{\lambda}^*_t\sum_il_i\sum_{s=\max\{1,t-\tau_i+1\}}^t\hat{x}^*_{is} = \sum_t\hat{\lambda}^*_t(\hat{q}^*_t+g_t)= \sum_t\overline{p}^{\text{gen}*}_t(\overline{q}^*_t+g_t)
.\end{split}\end{equation}
From the definition of $p^{\hat{\nu}^*}_{it}$ and the flexibility constraints (\ref{SPPconst3}) and (\ref{SPPconst4}), the right term in (\ref{budget_bal_proof1}) is equal to 
\begin{equation}\label{budget_bal_proof3}\begin{split}&\sum_i\sum_t\hat{\nu}^{S*}_{it}\sum_{s=1}^t\sum_{r=\max\{1,t-\tau_i+1\}}^s\hat{x}^*_{ir}+ \sum_i\sum_t\hat{\nu}^{E*}_{it}\sum_{s=t}^T\sum_{r=\max\{1,t-\tau_i+1\}}^s\hat{x}^*_{ir}\\
&\hspace{0.5in}= \sum_i\sum_t\tau_i\hat{\nu}^{S*}_{it}(1-\hat{y}^*_{it}) +\sum_i\sum_t\tau_i\hat{\nu}^{E*}_{it}(1-\hat{z}^*_{it})\\
&\hspace{0.5in}= \sum_i\sum_t\overline{p}^{S*}_{it}(1-\overline{y}^*_{it}) +\sum_i\sum_t\overline{p}^{E*}_{it}(1-\overline{z}^*_{it}).\end{split}\end{equation}
Summing the last expressions in (\ref{budget_bal_proof2}) and (\ref{budget_bal_proof3}) gives the left hand side of (\ref{budget_bal_def}), showing that budget balance holds at the competitive equilibrium. 
Finally, the mechanism is efficient by Theorem 4, as it randomly activates loads of type $i$ according to $\hat{x}^*_{it}$. 
\end{proof}
\section{Two-Bus Model Optimality Conditions}\label{twobus_appendix}

In addition to feasibility, the KKT conditions for the two-bus, centralized social planner's problem are
\begin{align}
c'(\hat{q}^*_t)-\hat{\mu}^*_t&\geq 0\quad\forall\,t\\
\hat{q}^*_t(c'(\hat{q}^*_t)-\hat{\mu}^*_t)&= 0\quad\forall\,t\\
-\overline{U}_i + p^{\hat{\lambda}*}_{it} + p^{\hat{\nu}*}_{it} &\geq 0\quad\forall\,i,\,t\leq T-\tau_i+1\\
\hat{x}^*_{it}(-\overline{U}_i + p^{\hat{\lambda}*}_{it} + p^{\hat{\nu}*}_{it}) &= 0\quad\forall\,i,\,t\leq T-\tau_i+1\\
p^{\hat{\lambda}^*}_{it}+p^{\hat{\nu}*}_{it}&\geq 0\quad\forall i,\,t>T-\tau_i+1\\
\hat{x}^*_{it}(p^{\hat{\lambda}^*}_{it}+p^{\hat{\nu}*}_{it})&= 0\quad\forall i,\,t>T-\tau_i+1\\
\tau_i\hat{\nu}^{S*}_{it} - u^{dS}_{it}&\geq 0\quad\forall\,i,\,t\\
\hat{y}^*_{it}(\tau_i\hat{\nu}^{S*}_{it} - u^{dS}_{it})&=0\quad\forall\,i,\,t\\
\tau_i\hat{\nu}^{E*}_{it} - u^{dE}_{it}&\geq 0\quad\forall\,i,\,t\\
\hat{z}^*_{it}(\tau_i\hat{\nu}^{E*}_{it} - u^{dE}_{it})&=0\quad\forall\,i,\,t\\
\label{theta_KKT}\hat{\mu}^*_t - \hat{\lambda}^*_t +\hat{\gamma}^*_{21,t}&= 0\quad\forall\,t\\
\hat{\lambda}^*_t\left(-B\hat{\theta}^*_{2,t} - g_t + \sum_il_i\sum_{s=\max\{1,t-\tau_i+1\}}^t\hat{x}_{is}\right)&=0\quad\forall\,t
\end{align}
\begin{align}
\hat{\lambda}^*_t&\geq 0\quad\forall\,t.
\end{align}

The aggregated constraints from all individual problems can now be written as :
\begin{align}
c'(q^{G*}_t) - P_{2,t}&\geq 0\quad \forall\,t\\
q^{G*}_t\left(c'(q^{G*}_t) - P_{2,t}\right)&=0\quad\forall\,t\\
-\overline{U}_i + P^{\text{con}}_{it} + P^{\zeta^*}_{it}&\geq 0\quad\forall\,t\leq T-\tau_i+1\\
x^{C*}_{it}(-\overline{U}_i + P^{\text{con}}_{it} + P^{\zeta^*}_{it})&= 0\quad\forall\,t\leq T-\tau_i+1
\end{align}
\begin{align}
P^{\text{con}}_{it} + P^{\zeta^*}_{it}&\geq 0\quad\forall\,t>T-\tau_i+1\\
x^{C^*}_{it}\left(P^{\text{con}}_{it} + P^{\zeta^*}_{it}\right)&= 0\quad\forall\,t>T-\tau_i+1\\
P^S_{it} - u^{dS}_{it}+\tau_i\zeta^{S*}_{it}&\geq 0\quad\forall\,t\\
y^{C*}_{it}\left(P^S_{it} - u^{dS}_{it}+\tau_i\zeta^{S*}_{it}\right)&=0\quad\forall\,t\\
P^E_{it} - u^{dE}_{it}+\tau_i\zeta^{E*}_{it}&\geq 0\quad\forall\,t\\
z^{C*}_{it}\left(P^E_{it} - u^{dE}_{it}+\tau_i\zeta^{E*}_{it}\right)&=0\quad\forall\,t\\
P_{2,t}-\beta^*_t&\geq 0\quad\forall\,t\\
q^{I*}_{t}\left(P_{2,t}-\beta^*_t\right)&=0\quad\forall\,t\\
-p^{P_1}_{1,t} + p^{\alpha*}_t &\geq 0\quad \forall\,t\\
x^{I*}_{it}\left(-p^{P_1}_{1,t} + p^{\alpha*}_t\right)&=0\quad\forall\,t\\
-\alpha^*_t + \beta^*_t - \xi^*_{12,t} + \xi^*_{21,t}&=0\quad\forall\,t\\
\alpha^*_t\left(-B\theta^{I*}_{2,t} - g_t + \sum_il_i\sum_{s=\max\{1,t-\tau_i+1\}}x^{I*}_{it}\right)&=0\quad\forall\,t\\
\alpha^*_t&\geq0\quad\forall\,t
\end{align}
Comparison of these optimality conditions yields adapted versions of the proofs for Theorems 1 and 2. 

\end{document}